\documentclass[a4paper, oneside, reqno, 11pt]{amsart}

\usepackage{amssymb,amsmath,amsthm,amsfonts,amscd, mathtools}
\usepackage{amsaddr}
\usepackage{mathdots}
\usepackage{cases}
\usepackage{graphicx,psfrag,paralist}
\usepackage{microtype}
\usepackage{xspace}
\usepackage{thmtools}
\usepackage{marvosym}
\usepackage{tabularx}
\usepackage{csquotes}
\usepackage{color,soul}
\usepackage{subcaption}
\usepackage{algorithmic}
\usepackage{tikz,ifthen,calc}
\usepackage[a4paper, margin=2.7cm]{geometry}
\usepackage[english]{babel}
\usepackage[normalem]{ulem}
\usepackage[hyperfootnotes=false, colorlinks, citecolor= blue, urlcolor=blue, linkcolor=blue ]{hyperref}
\usepackage{verbatim}
\usepackage[linesnumbered, ruled, vlined, algosection]{algorithm2e}
\usepackage{calc}
\usepackage{dsfont}
\usepackage{tablefootnote}
\usepackage{multirow}
\usepackage{footmisc}

\usepackage[linesnumbered, ruled]{algorithm2e}
\SetKwInput{KwRuleI}{Rule I}
\SetKwInput{KwComplexity}{Complexity}
\SetKw{Continue}{continue}
\SetKw{Break}{break}
\SetKw{False}{false}
\SetKw{True}{true}
\SetKw{And}{and}
\SetKw{Or}{or}
\SetKw{Delete}{delete}
\SetKw{Not}{not}
\SetKw{Then}{then}

\usepackage[T1]{fontenc}
\usepackage{lmodern}

\usepackage{tabularray}
\UseTblrLibrary{siunitx}
\UseTblrLibrary{counter}
\UseTblrLibrary{functional}

\DefTblrTemplate{caption-tag}{caption-ams-style}{\textsc{Table~\thetable .}}
\DefTblrTemplate{caption-sep}{caption-ams-style}{\enskip}
\DefTblrTemplate{caption-text}{caption-ams-style}{\InsertTblrText{caption}}
\DefTblrTemplate{conthead-text}{caption-ams-style}{(Continued~on~next~page)}

\ExplSyntaxOn
\DefTblrTemplate{caption}{caption-ams-style}{\vspace{0.7em}
	\hbox_set:Nn \l__tblr_caption_box
	{
		\UseTblrTemplate { caption-tag } { default }
		\UseTblrTemplate { caption-sep } { default }
		\UseTblrTemplate { caption-text } { default }
	}
	\dim_compare:nNnTF { \box_wd:N \l__tblr_caption_box } > { 0.85\linewidth }
	{
		\noindent
		\makebox[\linewidth][c]{
			\parbox{0.85\textwidth}{
			\hbox_unpack:N \l__tblr_caption_box
			}
		}
		\par
	}
	{
		\centering
		\makebox [\linewidth] [c] { \box_use:N \l__tblr_caption_box }
		\par
	}
}
\DefTblrTemplate{capcont}{caption-ams-style}{\vspace{0.7em}
	\hbox_set:Nn \l__tblr_caption_box
	{
		\UseTblrTemplate { caption-tag } { default }
		\UseTblrTemplate { caption-sep } { default }
		\UseTblrTemplate { caption-text } { default }
		\UseTblrTemplate { conthead-pre } { default }
		\UseTblrTemplate { conthead-text } { default }
	}
	\dim_compare:nNnTF { \box_wd:N \l__tblr_caption_box } > { 0.85\linewidth }
	{
		\noindent
		\makebox[\linewidth][c]{
			\parbox{0.85\textwidth}{
			\hbox_unpack:N \l__tblr_caption_box
			}
		}
		\par
	}
	{
		\centering
		\makebox [\linewidth] [c] { \box_use:N \l__tblr_caption_box }
		\par
	}
}
\ExplSyntaxOff

\DefTblrTemplate{firsthead,middlehead,lasthead}{caption-bottom-ams-style}{
}
\DefTblrTemplate{firstfoot}{caption-bottom-ams-style}{
	\UseTblrTemplate{capcont}{caption-ams-style}
}
\DefTblrTemplate{middlefoot}{caption-bottom-ams-style}{
	\UseTblrTemplate{capcont}{caption-ams-style}
}
\ExplSyntaxOn
\DefTblrTemplate{lastfoot}{caption-bottom-ams-style}{
	\UseTblrTemplate{note}{default}
	\UseTblrTemplate{remark}{default}
	\int_compare:nNnTF { \l__tblr_table_page_int } = {1} {
		\UseTblrTemplate{caption}{caption-ams-style}
	} {
		\UseTblrTemplate{capcont}{caption-ams-style}
	}
}
\ExplSyntaxOff

\NewTblrTheme{ams-theme}{
	\SetTblrTemplate{caption-tag}{caption-ams-style}
	\SetTblrTemplate{caption-sep}{caption-ams-style}
	\SetTblrTemplate{caption-text}{caption-ams-style}

	\SetTblrTemplate{conthead-text}{caption-ams-style}
	\SetTblrTemplate{caption}{caption-ams-style}
	
	\SetTblrTemplate{firsthead,middlehead,lasthead,firstfoot,middlefoot,lastfoot}{caption-bottom-ams-style}
}

\usepackage{chngcntr}
\usepackage{etoolbox} \patchcmd{\section}{\scshape}{\bfseries}{}{} \makeatletter \renewcommand{\@secnumfont}{\bfseries} \makeatother
\newrobustcmd\TableBold{\DeclareFontSeriesDefault[rm]{bf}{b}\bfseries}
\usetikzlibrary{positioning}
\usetikzlibrary{shapes}
\usetikzlibrary{shapes.symbols,patterns}
\usetikzlibrary{calc,through,backgrounds}
\usetikzlibrary{decorations.pathreplacing,calligraphy}

\usepackage{placeins}

\let\Oldsection\section
\renewcommand{\section}{\FloatBarrier\Oldsection}

\let\Oldsubsection\subsection
\renewcommand{\subsection}{\FloatBarrier\Oldsubsection}

\usepackage{fewerfloatpages}

\newtheorem{theorem}{Theorem}[section]
\theoremstyle{plain}

\newtheorem{proposition}[theorem]{Proposition}

\newtheorem{definition}[theorem]{Definition}

\newtheorem*{fact*}{Fact}
\newtheorem*{approach*}{Approach}
\newtheorem*{obstacle*}{Obstacle}
\newtheorem*{note*}{Note}

\theoremstyle{remark}

\newcommand{\nin}{\not \in}
\renewcommand{\geq}{\geqslant} \renewcommand{\leq}{\leqslant} 

\numberwithin{theorem}{section}
\numberwithin{equation}{section}
\counterwithin{table}{section}
\counterwithin{figure}{section}
\counterwithin{algocf}{section}

\newcommand{\sm}{\left(\begin{smallmatrix}} \newcommand{\esm}{\end{smallmatrix}\right)} \newcommand{\bpm}{\begin{pmatrix}} \newcommand{\ebpm}{\end{pmatrix}}

\newcommand{\rquotient}[2]{
	\mathchoice
	{% \displaystyle
		\text{\raise01ex\hbox{$#1$}\Big/\lower1ex\hbox{$#2$}}%
	}
	{% \textstyle
		#1\,/\,#2
	}
	{% \scriptstyle
		#1\,/\,#2
	}
	{% \scriptscriptstyle  
		#1\,/\,#2
	}
}

\newcommand{\bigrquotient}[2]{
	\mathchoice
	{% \displaystyle
		\text{\raise01ex\hbox{$#1$}\Bigg/\lower1ex\hbox{$#2$}}%
	}
	{% \textstyle
		#1\,/\,#2
	}
	{% \scriptstyle
		#1\,/\,#2
	}
	{% \scriptscriptstyle  
		#1\,/\,#2
	}
}

\newcommand{\sslash}{\slash \mkern-5.5mu \slash}

\def\deg{\rm{deg}}

\newcommand{\cpp}{{C\nolinebreak[4]\hspace{-.05em}\raisebox{.4ex}{\tiny \textbf{++}}}}

\usepackage{appendix}
\newtheorem{manualtheoreminner}{Theorem}
\newenvironment{restatethm}[1]{%
	\IfBlankTF{#1}
	{\renewcommand{\themanualtheoreminner}{\unskip}}
	{\renewcommand\themanualtheoreminner{#1}}%
	\manualtheoreminner
}{\endmanualtheoreminner}

\author{\vspace{-1.5em} \footnotesize Toni B\"ohnlein\textsuperscript{\textdagger}, P\'al Andr\'as Papp\textsuperscript{\textdagger}, Raphael S. Steiner\textsuperscript{\textdagger},\\ Christos K. Matzoros, and Albert-Jan N. Yzelman}

\address{\vspace{-0.6em}\texttt{\footnotesize\emph{\{toni.boehnlein; pal.andras.papp; raphael.steiner; albertjan.yzelman\}@huawei.com\\ \vspace{-0.2em} christos.konstantinos.matzoros@h-partners.com}} \\ \vspace{0.9em}
		\footnotesize Huawei Research Center Zurich, Computing Systems Lab,\\ Thurgauerstrasse 80, 8050 Zurich, Switzerland
}

\date{\today}

\title{Efficient parallel scheduling for\\ sparse triangular solvers}

\begin{document}
	
\begin{abstract}
We develop and analyze new scheduling algorithms for solving sparse triangular linear systems (SpTRSV) in parallel. Our approach produces highly efficient synchronous schedules for the forward- and backward-substitution algorithm. Compared to state-of-the-art baselines HDagg \cite{zarebavani2022hdagg} and SpMP \cite{park2014sparsifying}, we achieve a $3.32 \times$ and $1.42 \times$ geometric-mean speed-up, respectively. We achieve this by obtaining an up to $12.07 \times$ geometric-mean reduction in the number of synchronization barriers over HDagg, whilst maintaining a balanced workload, and by applying a matrix reordering step for locality. We show that our improvements are consistent across a variety of input matrices and hardware architectures.
\end{abstract}

\maketitle

\renewcommand*{\thefootnote}{\fnsymbol{footnote}}
\footnotetext[2]{Joint first authors; listed in alphabetical order.}
\renewcommand*{\thefootnote}{\arabic{footnote}}

\setcounter{tocdepth}{1}
\tableofcontents

\section{Introduction}
\label{sec:intro}

Systems of linear equations are ubiquitous and solving them fast numerically with high accuracy is essential to engineering, big data analytics, artificial intelligence, and various scientific fields. 
Key techniques in scaling to ever larger linear systems have been exploiting the sparsity of non-zero coefficients in modern algorithms, as well as leveraging the multi-core or multi-processor architectures of high-performance computing systems. 
However, whilst sparsity reduces computational load, the typically irregular distribution of non-zero elements complicates the development of efficient parallel algorithms, as the lack of structure hinders workload balancing and limits the ability to minimize communication between processors.

In this paper, we concern ourselves with solving sparse triangular systems of linear equations (SpTRSV) using parallel machines; i.e., solving a linear system $Lx = b$, where $L$ is a sparse triangular matrix and $b$ is a dense vector. 
Although solving sparse triangular linear systems marks a special case, it often arises as an important step in procedures solving more general linear systems. Some concrete examples are (sparse) LU, QR, and Cholesky decompositions, Gau{\ss}--Seidel, and so forth.
Efficient parallel-computation schedules for SpTRSV are of particular importance in applications where the same sparsity pattern is used repeatedly.
Such is the case in  simulations of various physical systems, for instance, ones that are based on the finite element method on a fixed mesh.

One of the main methods of solving SpTRSV is the forward-/backward-substitution algorithm.  
An execution of the algorithm on an instance may be captured by a directed acyclic graph (DAG), with the vertices corresponding to the rows of the matrix and directed edges representing dependencies imposed by the non-zero entries, see Figure~\ref{fig:example_dag}.
Finding a parallel execution of the forward-/backward-substitution algorithm directly corresponds to solving the parallel-scheduling problem on the corresponding DAG.

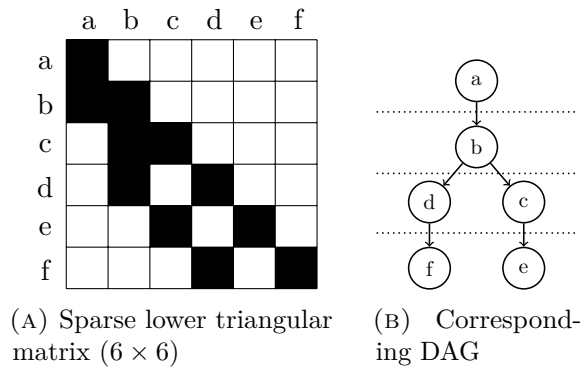
\begin{figure}[!htp]
	\centering
	%	\hfill
    \hspace{-0.01\textwidth}
	\begin{subfigure}[b]{0.27\textwidth}
		\centering
		\begin{tikzpicture}[scale=0.55] % Adjust scale for better fit
			% Draw the 6x6 grid
			\draw[step=1,black,thin] (0,0) grid (6,6);
			
			% Fill some entries in the lower triangular part
			\foreach \i/\j in {1/1, 2/1, 2/2, 3/2, 3/3, 4/2, 4/4, 5/3, 5/5, 6/4, 6/6} {
				\fill[black] (\j-1, 6-\i) rectangle (\j, 7-\i);
			}
			
			% Add row numbers
			\foreach \i/\l in {1/a, 2/b, 3/c, 4/d, 5/e, 6/f} {
				\node[align=left] at (-0.5, 6-\i+0.45) {\l};
			}
		
			% Add column numbers
			\foreach \i/\l in {1/a, 2/b, 3/c, 4/d, 5/e, 6/f} {
				\node[anchor=base] at (-0.9+\i+0.45, 6.25) {\l};
			}
		\end{tikzpicture}
		\caption{Sparse lower triangular matrix ($6\times 6$)} \label{sfig:small-matrix}
	\end{subfigure}
    \hspace{0.02\textwidth}
	\begin{subfigure}[b]{0.17\textwidth}
		\centering
		\resizebox{1.0\textwidth}{!}{\begin{tikzpicture}[scale=0.6, ->, auto, node distance=1.2cm, thick, main node/.style={circle, draw, fill=white, minimum size=7.5mm}]
			
			% Define nodes in a compact, hierarchical layout
			\node[main node] (a) {a};
			\node[main node] (b) [below of=a, yshift=-0.0cm, xshift=-0cm] {b};
			\node[main node] (c) [below right of=b, yshift=-0.15cm, xshift=0cm] {c};
			\node[main node] (d) [below left of=b, yshift=-0.15cm, xshift=-0cm] {d};
			\node[main node] (f) [below of=d, yshift=-0.0cm, xshift=0cm] {f};
			\node[main node] (e) [below of=c, yshift=-0.0cm, xshift=0cm] {e};
			
			\draw[dotted, -] (-3,-0.94)--(3,-0.94) ;
			\draw[dotted, -] (-3,-2.8)--(3,-2.8);
			\draw[dotted, -] (-3,-4.6)--(3,-4.6);
			% Draw edges corresponding to the non-zero entries in the matrix with reversed direction
			\path[every node/.style={font=\sffamily\small}]
			
			(a) edge (b)
			(b) edge (c)
			(b) edge (d)
			(c) edge (e)
			(d) edge (f);
		\end{tikzpicture}}
		\caption{Corresponding DAG} \label{sfig:small-dag}
	\end{subfigure}
    \hspace{-0.02\textwidth}
	%	\hfill
	\caption{A sparse lower triangular matrix \hyperref[sfig:small-matrix]{(a)} and its corresponding DAG for the forward-substitution algorithm \hyperref[sfig:small-dag]{(b)}. Each row of the matrix corresponds to a vertex in the DAG. An edge from vertex \(u\) to vertex \(v\) exists if and only if there is a non-zero entry in column \(u\) of row \(v\) in the matrix. The dotted lines in Figure \hyperref[sfig:small-dag]{(b)} separate the wavefronts of the DAG.}
	\label{fig:example_dag}
\end{figure}

In order to generate an efficient parallel schedule for the algorithm, one needs to:
\begin{enumerate}
	\item balance workload across machines, and \label{enum:into-parallel-sptrsv-criteria-work-balance}
	\item limit coordination overhead. \label{enum:into-parallel-sptrsv-criteria-coordination-overhead}
\end{enumerate}

Satisfying both of these needs simultaneously has proven to be challenging due to the irregular interdependence of computed values and the fine-grained nature of the problem.
Early algorithms include so-called wavefront schedulers \cite{anderson1989solving,saltz1990aggregation}, which repeatedly schedule all computations whose prerequisites are met, known as the wavefronts, cf.\@ Figure~\ref{sfig:small-dag}, followed by a synchronization barrier. 
They, however, suffer from large overhead stemming from frequent global synchronization \cite{park2014sparsifying}.
Similarly, early asynchronous approaches such as self-scheduling \cite{saltz1988run} had the drawback of incurring overheads due to numerous fine-grained synchronizations \cite{rothberg1992parallel}.

In a breakthrough paper, Park \emph{et al.\@} \cite{park2014sparsifying} reduced coordination overhead by combining these earlier ideas. Their scheduler SpMP, which remains a competitive baseline to date, is in essence an asynchronous wavefront scheduler: it allows machines to move onto the next wavefront if and only if all requisites have already been met for its portion of the next wavefront. 
They also developed a fast approximate transitive reduction to reduce the number
of synchronization points further. An alternate reduction in synchronizations has been made by Yilmaz \emph{et al.\@} \cite{yilmaz2020adaptive} by enforcing a bound by which machines may be out of sync. 

For synchronous schedulers, efforts have been directed towards increasing the computational load between synchronization barriers, thus decreasing the number of global synchronizations. For instance, Cheshmi \emph{et al.\@} \cite{cheshmi2018parsy} devise such methods for triangular matrices of a special structure, arising in Cholesky decompositions. For general sparse triangular matrices, a state-of-the-art baseline is the recent scheduler HDagg of Zarebavani \emph{et al.\@} \cite{zarebavani2022hdagg}. This algorithm develops efficient schedules by gluing together consecutive wavefronts if and only if a balanced workload can still be maintained and by pre-applying a DAG coarsening technique.

\subsection{Our contribution}

Our work continues along the same path of reducing the number of synchronization barriers. We present and analyze a new scheduling algorithm named \emph{GrowLocal} which is tailored specifically towards the SpTRSV application. In our experiments, we establish that this algorithm produces significantly superior parallel schedules compared to the baseline methods. Specifically, GrowLocal achieves a reduction in execution time of $1.42 \times$ compared to SpMP and of $3.32 \times$ compared to HDagg, on the SuiteSparse Matrix Collection benchmark \cite{davis2011university} used by previous studies, see Figure~\ref{fig:intro_plot}. 
We further evaluate GrowLocal on pre-processed variants of the SuiteSparse matrices motivated by applications, and observe speedups of up to $1.80 \times$ and $2.20 \times$ over SpMP and HDagg, respectively. 
On synthetic random matrices that are hard to schedule by design, the difference to the baselines is even larger: the algorithm achieves a speed-up of $2.50 \times$ compared to SpMP and $10.12 \times$ compared to HDagg in execution time.

\begin{figure}[!htp]
	\centering
	\includegraphics[scale=0.55]{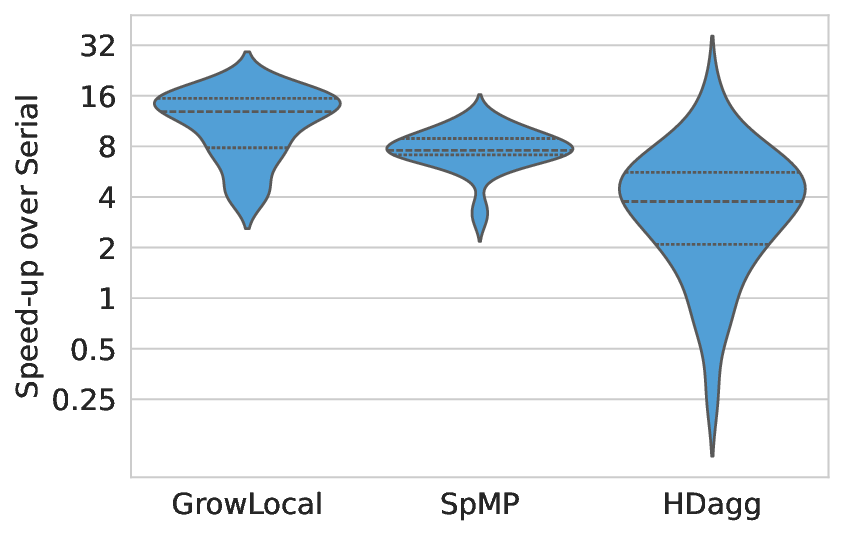}
	\captionof{figure}{Geometric mean and interquartile ranges of speed-ups over Serial of our algorithms on the SuiteSparse Matrix Collection \cite{davis2011university} on an Intel x86 machine using 22 cores.}
	\label{fig:intro_plot}
\end{figure}

The algorithm obtains these speed-ups by significantly reducing the number of synchronization barriers required: we report a $12.07 \times$ reduction in the number of barriers relative to HDagg on the SuiteSparse data set, whilst maintaining a good workload balance. The results also show that our scheduler provides consistent improvements over several different computing architectures and types of input matrices. The running time of the scheduling algorithm itself is also comparable to the state-of-the-art baselines, making it a viable tool for various applications.

In summary, the main contributions of our paper are:
\begin{itemize}
	\item a novel algorithm for generating efficient parallel schedules for SpTRSV execution;
	\item extensions of previous DAG coarsening techniques that enhance the schedules, and a short theoretical proof that these preserve acyclicity; and
	\item experiments confirming that the above schedulers achieve significant speed-ups over the SpMP and HDagg baselines, on various architectures and data sets, including an ablation study of the individual techniques proposed.
\end{itemize}

\subsubsection{GrowLocal scheduling algorithm} \label{sec:intro-list-scheduler-with-barriers}

There are numerous prior works on parallel DAG scheduling in the literature. When the number of cores is limited, the best results are often achieved by so-called list scheduling algorithms~\cite{graham1969bounds,adam1974comparison, hwang1989scheduling, radulescu2002low, mingsheng2003efficient}, which schedule the vertices in a topological order according to some priority function~\cite{wang2018list}. On the other hand, DAG scheduling with barrier synchronization is a somewhat different setting, and there are only a few previous works that address this problem. One state-of-the-art example here is the HDagg algorithm mentioned before~\cite{zarebavani2022hdagg}, which can also be interpreted as scheduler for general DAGs. Besides this, the idea of adapting list schedulers to a barrier synchronization setting has also been explored recently by Papp \emph{et al.}~\cite{Papp2024Efficient} for abstract bulk-synchronous-parallel (BSP) scheduling.

Our GrowLocal algorithm takes a rather different approach than these previous methods, but also incorporates some of their underlying strengths. On a high level, GrowLocal considers a parameter $\alpha$, and tries to form the part of the schedule until the next synchronization barrier (the next so-called \emph{superstep}) by assigning approximately $\alpha$ vertices to each of the cores before this next barrier. The parameter $\alpha$ is then iteratively increased, examining larger parts of the DAG as the potential next superstep, as long as this is possible while also ensuring a sufficient amount of parallelization between the cores.

During the development of the schedule, the algorithm always maintains the set of vertices that are ready to be executed, i.e., all their parents have been computed. At any point during the algorithm, if we consider the current superstep, such a ready vertex $v$ may be executable either on any of the cores (if all parents of $v$ were computed before the last barrier), only on a specific core $p$ (if a parent of $v$ was computed on $p$ since the last barrier), or on none of the cores (if we computed parents of $v$ on multiple cores since the last barrier). When selecting the next vertices to assign to a core $p$ during our algorithm, GrowLocal first prioritizes vertices that can only be executed on the core $p$ before the next barrier. This is inspired by the scheduler of~\cite{Papp2024Efficient}, and it ensures that we can compute significantly more vertices before having to insert a new barrier.

Apart from this, GrowLocal simply selects vertices to assign to a core based on their IDs in order to group neighboring vertices onto the same core and superstep. 
This leads to significantly better locality for the developed schedule than in case of, e.g., list schedulers, and this has a large positive impact on the overall performance of the SpTRSV computation.

\subsubsection{Coarsening} \label{sec:intro-coarsening}
Graph coarsening techniques are widely applied in graph partitioning tools \cite{ishiga1975logic, karypis1997multilevel, Schlag2020_1000105953}, where they greatly reduce the size of the graph and improve data locality. These techniques can also be applied to DAG scheduling \cite{popp2021multilevel, zarebavani2022hdagg}, where they can further help to reduce the number of synchronization steps on top of the aforementioned benefits.
Following the coarsening, the scheduling algorithm is applied to the coarse graph and the resulting schedule is subsequently pulled back to the original graph to obtain the final schedule.
In order to produce a valid scheduling problem, the coarsening methods are required to preserve the acyclicity of DAGs. Methods that fulfill this criteria have been studied in several works before, see, for example, \cite{cong1994acyclic, fauzia2013beyond, herrmann2017acyclic, zarebavani2022hdagg} and references therein. 

In Section \ref{sec:acyclic-coarsening}, we introduce the concept of \emph{cascades} to generalize the coarsening techniques utilized in \cite[\S 4]{cong1994acyclic} and \cite[\S IV.B]{zarebavani2022hdagg}.
We then formally prove that coarsening techniques based on cascades always preserve acyclicity.
In Section \ref{sec:eval-funnel}, we evaluate the effect of the coarsening algorithm developed in Section \ref{sec:acyclic-coarsening} on our scheduling algorithm, GrowLocal.

\subsubsection{Reordering} \label{sec:intro-reordering} Besides the algorithms above, we also apply a matrix reordering step to drastically improve data locality during the SpTRSV computation. Specifically, once the schedule is developed, we symmetrically permute the matrix according to the schedule, ensuring that values computed after each other on the same core are close to each other in this permuted representation. This idea has already been explored by Rothberg--Gupta in the 1990s \cite{rothberg1992parallel}, but it has not been applied in modern SpTRSV baselines, which instead try to make use of existing data locality when deriving a schedule.

\subsubsection{Block parallel scheduling} \label{sec:intro-parallel-scheduling}

A known optimization technique for parallel SpTRSV execution is to break the lower triangular matrix into blocks \cite{anderson1989solving, mayer2009parallel, ahmad2021split, yilmaz2020adaptive}. These blocks may be on the diagonal, which corresponds to a smaller instance of (sparse) triangular solve, or completely off the diagonal, which corresponds to a (sparse) matrix-vector multiplication. The separation of the easily parallelizable (sparse) matrix-vector-multiplication blocks from the hard to parallize (sparse) triangular blocks has been particularly impactful for GPU implementations \cite{liu2016synchronization, lu2020efficient}.

In this paper, we use this block decomposition to run the GrowLocal scheduling algorithm in parallel on each (sparse) triangular block. The resulting synchronous schedules can then be combined one after the other (with a synchronization barrier between individual schedules) to a schedule for the whole triangular matrix. This leads to a super-linear speed-up in the scheduling time whilst having a moderate effect on the parallel SpTRSV solve time.

\subsection{Additional related work} \label{section:additional_work}

Besides the forward-/backward-substitution algorithm, there are also other methods for solving sparse triangular systems, for example inversion. For this method, we mention the memory-optimal algorithms developed in previous works \cite{alvarado1993optimal, pothen1992fast}.

\section{Background}
\label{sec:prelim}

\subsection{Graph notation} We model our computations as a directed acyclic graph (DAG) $G=(V,E)$, which consists of a set of vertices $V$ and a set of directed edges $E \subseteq V \times V$. For any vertex $v \in V$, the sets of vertices $\{u \, | \, (u,v) \in E \}$ and $\{u \, | \, (v,u) \in E \}$ are called the \emph{parents} of $v$ and the \emph{children} of $v$, respectively. The \emph{in-} and \emph{out-degree} of $v$, denoted by ${\rm deg}^-(v)$ and ${\rm deg}^+(v)$, respectively, are the number of parents and children of $v$. The \emph{degree} of $v$, denoted by ${\rm deg}(v)$, is the sum of its in- and out-degree. If a vertex of the DAG has no parents/children, then it is called a \emph{source/sink} vertex, respectively. The DAG in our model is also complemented by vertex weights $\omega : V \rightarrow \mathbb{Z}_{>0}$ to indicate the compute cost of each operation.
\subsection{Problem definition and notation}
\label{sec:problem-description-and-forward-backward-substitution-algorithm}

When solving sparse triangular systems, we are given a triangular matrix $A = (A_{i,j})_{i,j=1,\dots,n} \in \mathbb{R}^{n \times n}$, a dense vector $b = (b_1,\dots,b_n)^T \in \mathbb{R}^n$, and the goal is to solve the equation $Ax=b$ for the vector $x = (x_1,\dots, x_n)^T \in \mathbb{R}^n$. We assume that $A$ is non-singular, such that all its diagonal elements are non-zero. In case of a lower triangular matrix $A$, there is a natural \textit{forward-substitution algorithm} for the problem, which iterates through the rows of $A$ in order and computes the values of $x$ as $x_1=\tfrac{b_1}{A_{1,1}}$, $x_2 = \tfrac{b_2 - A_{2,1} x_1}{A_{2,2}}$, and, in general, as
\begin{equation} \label{eq:forward-substitution-algorithm}
x_i = \frac{1}{A_{i,i}} \left(b_i - \sum_{j=1}^{i-1} A_{i,j}  x_j \right).
\end{equation}
In case of an upper triangular matrix $A$, a backward-substitution algorithm follows symmetrically in the reverse direction.

In the forward-substitution algorithm \eqref{eq:forward-substitution-algorithm}, 
we say that
the computation of $x_i$ \emph{depends} on the value of $x_j$, for $j<i$, if and only if there is an increasing sequence $j= \ell_0 < \ell_1 < \dots < \ell_m = i$ such that each entry $A_{\ell_{k-1},\ell_{k}}$ is non-zero, for $k=1,\dots,m$. If there is no dependency between $x_i$ and $x_j$, the two corresponding operations can be executed in any order, in particular also in parallel. As such, the operations in the algorithm can naturally be represented as a DAG $G=(V,E)$, where $V = \{1, ..., n\}$, the vertex $i$ represents the $i$-th row of $A$, and, for any $i,j \in V$, we have a directed edge $(j, i) \in E$ if and only if $A_{i,j} \neq 0$. See  Figure~\ref{fig:example_dag} for an example. To indicate the compute cost of each operation, the weight $\omega(v)$ of each vertex $v \in V$ in the DAG is simply defined as the number of non-zero entries in the corresponding row of the matrix.

The parallel execution of this DAG then directly corresponds to a parallel execution of the SpTRSV. Many previous works found it more convenient to discuss their scheduling methods for this problem using this DAG representation.

The parallel-scheduling problem above can be most fittingly captured in a bulk-synchronous parallel (BSP) model \cite{valiant1990bridging} that assumes \emph{global synchronization barriers} to split the execution into so-called \emph{supersteps}. This model is also known as the XPRAM model \cite{valiant1990general}. A schedule in this model assigns each vertex, i.e., the computation of each $x_i$, to one of the $k$ available cores and to a given superstep. A valid schedule must fulfill the precedence constraints of the DAG and ensure that we always have a synchronization barrier between computing a value on one core and using it as input on another core.

\begin{definition}
A \emph{parallel schedule} of $G$ consist of assignments $\pi : V \rightarrow \{1, ..., k\}$ to cores and $\sigma : V \rightarrow \mathbb{Z}_{>0}$ to supersteps, which fulfill the following properties for each $(u,v) \in E$:
\begin{itemize}
 \item $\sigma(u) \leq \sigma(v)$;
 \item if $\pi(u) \neq \pi(v)$, then $\sigma(u) < \sigma(v)$.
\end{itemize}
\end{definition}
The total cost of a schedule is determined by the workload balance within each superstep and the number of synchronization barriers. 
The original BSP model includes also communication volume in its cost function. For the SpTRSV application, however, the communication happens in parallel to the computation and resolving the synchronizations. Hence, the latter two dominate the overall execution time. 
Synchronous methods from previous works apply the same scheduling model, although often without explicitly referring to BSP, XPRAM, or supersteps.

\section{The GrowLocal scheduler}
\label{sec:algorithms}

Our GrowLocal algorithm is tailored specifically to the DAG scheduling problem with synchronization barriers. The algorithm forms the supersteps one by one, always aiming to make the current superstep as large as possible while maintaining a good workload balance. The current superstep is formed through several \emph{iterations} with a superstep length parameter $\alpha$. The algorithm attempts to form a new superstep with approximately $\alpha$ vertices assigned to each core, and gradually increases $\alpha$ as long as this allows sufficient parallelization.

Specifically, in a single iteration with parameter $\alpha$, the algorithm first assigns (up to) $\alpha$ vertices to the first core, and considers the sum $\Omega_1$ of the weights of these vertices. It then assigns vertices up to total weight of at most $\Omega_1$ to the second core, third core, and so forth. Let $\Omega_p$ denote the total weight allocated to core $p$ in this iteration. We associate a parallelization score of
\begin{equation} \label{eq:parallelization-score}
	\beta = \frac{\sum_p \Omega_p}{\max_p \Omega_p + L}
\end{equation}
to the current iteration. Here, $L$ is a parameter reflecting the penalty (time cost) incurred by each new synchronization barrier\footnote{The value of $L$ may be architecture dependent. In this study, we set $L=500$ based on synchronization cycles and a small empirical evaluation.}. In order to consider the superstep allocation of the current iteration \emph{worthy}, the algorithm requires that its score $\beta$ is relatively large, i.e., close to the parallelization score achieved in the previous iterations.

In order to form a superstep, the algorithm begins with a minimal length $\alpha = 20$ iteration. This first iteration is always considered worthy, regardless of its parallelization score. Then, in each subsequent iteration, we consider a different choice for the next superstep: the assignments of the previous iteration are invalidated, the parameter $\alpha$ is increased by a factor of $1.5$, and a new potential superstep is formed, with more vertices assigned to each core. If the resulting parallelization score is still high enough, then the superstep allocation of this iteration is also considered worthy, and the process continues. Otherwise, the last worthy superstep allocation is finalized as the current superstep. The high-level pseudocode of the algorithm is outlined in Algorithm~\ref{alg:GrowLocal}.

Naturally, when assigning vertices to a specific core $p$ in a superstep, there may be numerous ready-to-compute vertices that we can choose from, and selecting among these is a key aspect to any scheduler. Similarly to the heuristic of~\cite{Papp2024Efficient}, our algorithm first prioritizes those vertices that are only computable on $p$ in this superstep, since some of their parents were assigned to $p$ in the current superstep. In lack of such vertices, GrowLocal simply selects the vertices with smallest IDs.

We note that while this ID-based selection may seem simple at first, it in fact plays a crucial role in the success of our algorithm.  Previous scheduling heuristics usually assign vertices to the different cores simultaneously in order to ensure work balance. In contrast to this, GrowLocal first assigns vertices to the first core, then to the second core, and so forth. With the ID-based selection, this often leads to schedules where the vertices on a core are more-or-less consecutive blocks in the matrix, which drastically improves locality during the computation. This is especially important in matrices from applications, which are often already ordered superbly with respect to locality, and thus preserving this is crucial. As such, GrowLocal can also be loosely understood as a method combining the strengths of previous DAG schedulers with barrier synchronization: similarly to~\cite{Papp2024Efficient}, it allows a priority-based choice between vertices when forming a superstep, but as in~\cite{zarebavani2022hdagg}, it preserves locality by aiming to assign consecutive vertices to the same core.

\begin{algorithm}[!t]
	\DontPrintSemicolon
	\SetNlSty{textsc}{}{}
	\SetAlgoNlRelativeSize{-1}
	\caption{Skeleton of GrowLocal scheduler \label{alg:GrowLocal}}
	\KwData{A vertex-weighted DAG $G=(V,E,\omega)$ and a set of cores $P=\{1,2,\dots,k\}$.}
	\KwResult{A schedule consisting of processor assignment $\pi:V \to P$ and superstep assignment $\sigma: V \to \mathbb{Z}_{> 0}$.}
	\BlankLine
	\KwRuleI{Vertices are prioritized according to
		\begin{enumerate}[\hspace{1.6cm}(i)]
			\item core exclusivity, and then
			\item smallest ID.
	\end{enumerate}}
	
	\BlankLine
	\While{\emph{not all vertices are assigned yet}}{
		$\alpha \leftarrow 20$ \;
		\While{\True \label{alg-line:superstep-attempt}}{
			\smallskip
			\tcp{I. Assign new vertices to each core}
			assign up to $\alpha$ vertices to core $1$ with \textbf{Rule I}\;
			$\Omega_1 \leftarrow$ total newly assigned weight to core $1$\;
			\For{{\rm core} $p=2,\dots, k$ {\rm in} {\rm order}}{
				$\Omega_p \leftarrow 0$\;
				\While{$\Omega_p \not \approx \Omega_1$ \And \emph{can assign to core} $p$ \label{alg-line:GrowLocal-weight-balance}}{
					assign vertex $v$ to core $p$ with \textbf{Rule I}\;
					$\Omega_p \leftarrow \Omega_p + \omega(v)$\;
				}
			}
			
			\tcp{II. Check for sufficient parallelism}
			
			$\beta \leftarrow \frac{\sum_p \Omega_p}{\max_p \Omega_p + L}$ \;
			
			\eIf{\emph{the parallelization score $\beta$ is high enough}}{
				consider current assignment as worthy\;
				undo new assignments up to the last barrier \;
				$\alpha \leftarrow 1.5 \times \alpha$\;
			}
			{
				finalize last worthy assignment as next superstep\;
				\Break inner loop\;
			}
		}
		
	}
\end{algorithm}

Under reasonably mild assumptions, one can also show that the running time of the algorithm is almost linear.

\begin{theorem} \label{thm:GrowLocal-complexity}

Assume that both the out-degrees and compute weights of the DAG are on the same order of magnitude. Then, the time complexity of the GrowLocal algorithm is $O(|E| \cdot \log |V|)$.
\end{theorem}

The precise formulation of the theorem and the proof are deferred to Section~\ref{sec:growlocal-complexity} of the supplementary material. On a high level, since the size of the iterations follows a geometric series, one can show the total number of speculative vertex assignments in a superstep is still only a linear factor more than the size of the finalized superstep. However, a rigorous proof of the theorem is much more technical due to the fact that our parallelization score also depends on the weights of the vertices and the parameter $L$. In the supplementary material, we also provide a brief experimental analysis, which confirms linear complexity.

One can also easily observe that the space requirement of the algorithm is simply $O(|E|)$.

\subsection{Block parallel scheduling} \label{sec:growlocal-parallel-block-schedule}

In order to reduce the overhead from scheduling even further, one can parallelize the scheduling process. Instead of directly parallelizing the algorithm, we split the scheduling problem into independent parts. We achieve this by subdividing the lower triangular matrix into smaller lower triangular matrix blocks along the diagonal as in Figure~\ref{fig:block-decomp}. 

\begin{figure}[htpb]
	\centering
	%	\hfill
    \hspace{-0.02\textwidth}
	\begin{subfigure}[b]{0.27\textwidth}
		\centering
		\begin{tikzpicture}[scale=0.52] % Adjust scale for better fit
		% Draw the 9x9 grid
		\draw[step=1,black,very thin] (0,0) grid (9,9);
		
		% Draw Blocks
		\filldraw[fill=orange!40!white, fill opacity=0.4, draw=orange, line width=1.2pt] (0,6) rectangle (3,9);
		\filldraw[fill=blue!40!white, fill opacity=0.4, draw=blue, line width=1.2pt] (3,3) rectangle (6,6);
		\filldraw[fill=purple!40!white, fill opacity=0.4, draw=purple, line width=1.2pt] (6,0) rectangle (9,3);
		
		% Fill some entries in the lower triangular part
		\foreach \i/\j in {1/1, 2/1, 2/2, 3/2, 3/3, 4/2, 4/4, 5/3, 5/5, 6/4, 6/6, 7/7, 8/8, 9/9, 8/1, 8/3, 7/2, 7/3, 9/2, 9/7, 7/6} {
			\fill[black] (\j-1, 9-\i) rectangle (\j, 10-\i);
		}
		\end{tikzpicture}
		
		\caption{Decomposition of $9 \times 9$ matrix} \label{sfig:small-matrix2}
	\end{subfigure}
    \hspace{0.03\textwidth}
	\begin{subfigure}[b]{0.15\textwidth}
		\centering
		\resizebox{0.6\textwidth}{!}{\begin{tikzpicture}

    \node[anchor=center, rotate=90] at (-22pt,140pt) { DAG \large $1$};
    \node[anchor=center, rotate=90] at (-22pt,80pt) { DAG \large $2$};
    \node[anchor=center, rotate=90] at (-22pt,20pt) { DAG \large $3$};

    \draw[ultra thick, orange] (-12pt,168pt) rectangle (14pt,112pt);
    \draw[ultra thick, blue] (-12pt,108pt) rectangle (14pt,52pt);
    \draw[ultra thick, purple] (-12pt,48pt) rectangle (14pt,-8pt);

    \begin{scope}[thick, arrows=-stealth]
    \draw (0pt,160pt) -- (0pt,144pt);
    \draw (0pt,140pt) -- (0pt,124pt);
    \draw (0pt,100pt) -- (10pt,100pt) -- (10pt,60pt) -- (4pt,60pt);
    \draw (0pt,40pt) -- (10pt,40pt) -- (10pt,0pt) -- (4pt,0pt);
    \end{scope}

    \begin{scope}[very thick, densely dashed, arrows=-stealth]
    \draw (15pt,142pt) -- (28pt,142pt) -- (28pt,17pt) -- (15pt,17pt);
    \draw (15pt,138pt) -- (23pt,138pt) -- (23pt,82pt) -- (15pt,82pt);
    \draw (15pt,77pt) -- (23pt,77pt) -- (23pt,23pt) -- (15pt,23pt);
    \end{scope}

    \draw[black, fill=white] (0pt,0pt) circle (1.0ex);
    \draw[black, fill=white] (0pt,20pt) circle (1.0ex);
    \draw[black, fill=white] (0pt,40pt) circle (1.0ex);
    \draw[black, fill=white] (0pt,60pt) circle (1.0ex);
    \draw[black, fill=white] (0pt,80pt) circle (1.0ex);
    \draw[black, fill=white] (0pt,100pt) circle (1.0ex);
    \draw[black, fill=white] (0pt,120pt) circle (1.0ex);
    \draw[black, fill=white] (0pt,140pt) circle (1.0ex);
    \draw[black, fill=white] (0pt,160pt) circle (1.0ex);

\end{tikzpicture}}
		\caption{Sub-DAGs} \label{sfig:small-dag2}
	\end{subfigure}
    \hspace{-0.03\textwidth}
	%	\hfill
	\caption{Subdivision of a $9 \times 9$ lower triangular matrix into three $3 \times 3$ lower triangular matrix blocks \hyperref[sfig:small-matrix2]{(a)} and the three corresponding sub-DAGs with inter-DAG dependency \hyperref[sfig:small-dag2]{(b)}.}
	\label{fig:block-decomp}
\end{figure}
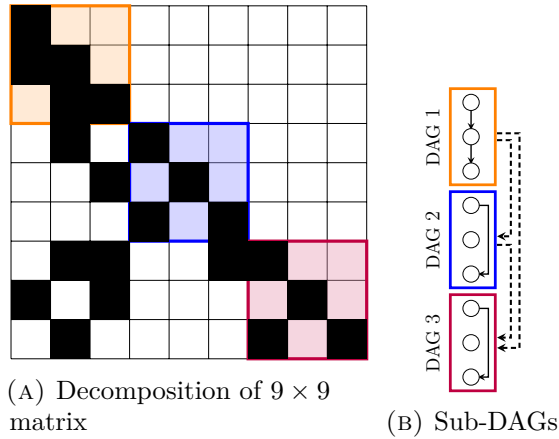

On each of these sub-problems, we can generate schedules in parallel. When combining these schedules, we just have to ensure that the individual schedules are combined one after the other. Equivalently, we can add to the superstep assignment of each vertex in each block the total number of supersteps in earlier block schedules.

We remark that for the weight of the vertices (in the DAG representation), we still use the number of non-zeros in the full matrix. This is in line with our SpTRSV kernel implementation.

\section{Acyclicity-preserving graph coarsening}
\label{sec:acyclic-coarsening}

Previous works discuss several ways to partition a DAG into clusters such that this coarsened graph remains acyclic, although often without a formal proof of this property. In a further generalization of earlier methods from Cong \emph{et al.} \cite[\S 4]{cong1994acyclic} and Zarebavani \emph{et al.} \cite[\S IV.B]{zarebavani2022hdagg}, we now introduce the concept of \emph{cascades} and prove that coarsening a DAG along such cascades is still guaranteed to preserve acyclicity. This is presented in Section \ref{sec:coarsen-generalities}. In Section \ref{sec:coarsen-algo}, we describe the graph coarsening algorithm used in our scheduling algorithms.

\subsection{Cascades} \label{sec:coarsen-generalities}

We begin with some formal definitions. Thereafter, we prove Proposition \ref{prop:cascade-coarsening}, demonstrating the utility of cascades for coarsening DAGs.

\begin{definition}
	Let $G=(V,E)$ be a directed graph and $P$ a partition of $V$. We define the \emph{coarsened graph} of $G$ along $P$ as the graph $(V',E')$, where $V'=P$, i.e., the vertices are the parts of the partition $P$, and for $U', W' \in V'$ we have that $(U',W') \in E'$ if and only if $U'\neq W'$ and $\exists (u,w) \in E$ such that $u \in U'$ and $w \in W'$. 
	We denote the coarsened graph of $G$ along $P$ by $G \sslash P$.
	\label{def:coarsened-graph}
\end{definition}

In other words, the coarsened graph $G \sslash P$ is the graph $G$ quotiented by the equivalence relation induced by $P$ with self-loops removed. The definition is easily extended to vertex-weighted graphs, where the weight of a part $U \in P$ is given as the sum the weights of its elements: $\omega(U)=\sum_{u \in U} \omega(u)$.

\begin{definition}
	Let $G=(V,E)$ be a directed graph. We call a subset of vertices $U \subseteq V$ a \emph{cascade} if and only if for every vertex $v \in U$ with an incoming cut edge, that is $(w,v) \in E$ such that $w \nin U$, and for every vertex $u \in U$ with an outgoing cut edge, that is $(u,w)\in E$ such that $w \nin U$, there is a (possibly trivial) directed walk from $v$ to $u$ in $G$.
	\label{def:cascade}
\end{definition}

\begin{proposition} Let $G=(V,E)$ be a directed acyclic graph and $P$ a partition of $V$ such that each set $U \in P$ is a cascade. Then, the coarsened graph $G \sslash P$ of $G$ along $P$ is acyclic.
	\label{prop:cascade-coarsening}
\end{proposition}

\begin{proof}
	We will show that any directed walk in $G \sslash P$ can be elevated to a directed walk in $G$. Therefore, the existence of directed cycles in $G \sslash P$ implies the existence of directed cycles in $G$.
	
	We lift a walk from $G \sslash P$ by mapping each edge to an arbitrary representative in $E$, whose endpoints necessarily lie in disjoint sets of the partition $P$ as $G \sslash P$ does not contain any self-loops, and connecting the endpoints via the directed walks guaranteed by the defining property of cascades.
\end{proof}

\subsection{Algorithm} \label{sec:coarsen-algo}

In our graph coarsening algorithm, we do not make use of the full strength of Proposition \ref{prop:cascade-coarsening}. Instead, we use a subcategory of cascades, which can be found efficiently. We call them \emph{funnels}, though they have been previously described under the name \emph{fanout-free cone}~\cite[\S 4]{cong1994acyclic}.  
Since the latter reference does not include an algorithm with a complexity analysis, we include them here in Algorithm \ref{alg:funnel}. 

\begin{definition} Let $G=(V,E)$ be a directed acyclic graph. We call a subset of vertices $U \subseteq V$ an \emph{in-funnel} if and only if $U$ is a cascade and there is at most one vertex $u \in U$ with an outgoing cut edge, that is $(u,w)\in E$ such that $w \nin U$.
	
	We analogously define an \emph{out-funnel}.
	\label{def:funnel}
\end{definition}

\begin{algorithm}[htpb]
	\DontPrintSemicolon
	\SetNlSty{textsc}{}{}
	\SetAlgoNlRelativeSize{-1}
	\caption{In-funnel graph coarsening.\label{alg:funnel}}
	\KwData{A directed acyclic graph $G=(V,E)$.}
	\KwResult{A partition $P$ such that every $U\in P$ is an in-funnel.}
	\BlankLine
	${\rm Partition } \gets \emptyset$\;
	${\rm Visited}[v] \gets \False, \, \forall v \in V$\;
	\For{$v \in V$ \emph{in reverse topological order}}{
		\lIf{${\rm Visited}[v]$}{\Continue}
		$U \leftarrow \emptyset$\;
		${\rm ChildrenCount}[u] \gets 0, \, \forall u \in V$\;
		${\rm PrioQueue}.{\rm insert}(v)$\;
		\While{\Not ${\rm PrioQueue}.{\rm empty}()$}{
			$w \leftarrow {\rm PrioQueue}.{\rm pop}()$\;
			$U.{\rm insert}(w)$\;
			\For{$u \in {\rm Parents}(w)$ \label{algline:funneldfs-parents} }{
				${\rm ChildrenCount}[u] \gets {\rm ChildrenCount}[u] + 1$\;
				\If{${\rm ChildrenCount}[u] = {\rm OutDegree}(u)$ \label{algline:funneldfs-children-check}}{
					${\rm PrioQueue}.{\rm insert}(u)$\;
				}
			}
		}
		\For{$u \in U$}{
			${\rm Visited}[u] \gets$ \True \;
		}
		${\rm Partition}.{\rm insert}(U)$\;
	}
	\Return {\rm Partition}\;
\end{algorithm}

The time complexity of the topological sort is $O(|V|+|E|)$ \cite{kahn1962topological} and its space complexity is $O(|V|)$. In order to bound the time complexity for the remaining part, we note that each parent vertex $v$ in Line \ref{algline:funneldfs-parents} is visited at most as many times as its out-degree, leading to an overall complexity of $O(|V|+|E|)$. The space complexity is easily seen to be $O(|V|)$.

In practice, before applying this graph coarsening, we remove some transitive edges from $G$ as this increases the likelihood of finding larger components. A complete transitive reduction is slow, though there are faster approximate transitive reductions, such as the `remove all long edges in triangles'-algorithm \cite[\S2.3]{park2014sparsifying} with a time complexity of $O(\sum_{v \in V} {\rm deg}(v)^2)$. This algorithm may be terminated early if a faster runtime is desired. In our implementation, we run the full (approximate) algorithm.

In our implementation, we also add a size/weight constraint on each component of the partition to Algorithm~\ref{alg:funnel} as otherwise a graph with only one sink vertex would be coarsened into a graph with only one vertex. Altogether, our coarsening algorithm Funnel is a more general\footnote{Every \emph{in-tree} is an \emph{in-funnel}.} and robust version of the coarsening algorithm in HDagg \cite{zarebavani2022hdagg}.

\section{Reordering for locality}
\label{sec:reorder}

Our algorithms already account for two of the most important factors in synchronous scheduling: work balance and the number of synchronization barriers. However, another major aspect that greatly influences the efficiency of a parallel SpTRSV execution is data locality, i.e., the number of required values that are already available in cache.

In order to address this, we apply a separate reordering step to ensure that vertices which are computed together are also stored together. The main idea of this approach has already been considered before, cf.~\cite{rothberg1992parallel}, but has not found its way into modern baselines. In particular, we consider a reordering (relabeling) the vertices of the input DAG based on the partitioning we developed, where we iterate through the supersteps in order, and within each superstep, we iterate through the cores in order. That is, we first start with the vertices $v$ with $\pi(v)=1$, $\sigma(v)=1$, then the vertices $v$ with $\pi(v)=2$, $\sigma(v)=1$, and so on, up to the vertices $v$ with $\pi(v)=k$, $\sigma(v)=1$, followed by vertices $v$ with $\pi(v)=1$, $\sigma(v)=2$, etc. Within a given core-superstep combination, we go through the vertices in the original order (which gives a topological ordering of the induced sub-DAG). 
We then symmetrically permute the input matrix and permute the right-hand-side vector of the SpTRSV problem accordingly.
Note that since the permutation provides a valid topological ordering of the vertices of the DAG, the resulting matrix is still lower triangular, resulting in an equivalent (symmetrically permuted) formulation of the SpTRSV problem.

We then execute the SpTRSV computation on the permuted problem, following our schedule, which ensures that vertices computed on the same core in the same superstep are stored close to each other, thus greatly improving locality during the computation.

\section{Experimental setup}
\label{sec:exp}

In this section, we present the experimental setup for the evaluation of our scheduling algorithm. Our implementations are available in the \emph{OneStopParallel} repository~\cite{OneStopParallel} on Github.

\subsection{Methodology}

For the evaluation, we used a standard SpTRSV implementation which iterates through the rows of the matrix which was stored in compressed sparse row (CSR) format \cite{tinney1967direct}. The algorithm was parallelized using the OpenMP library with the flags \texttt{OMP\_PROC\_BIND} and \texttt{OMP\_PLACES} set to \texttt{close} and \texttt{cores}, respectively.

We measured one hundred times the time it takes for a single SpTRSV execution using the chrono high-resolution clock. The measurements were taken whilst the system was `hot', meaning 
two untimed executions precede the timed executions. Between each SpTRSV execution, the right-hand-side vector $b$ was reset to all ones.
The experiments were repeated for each scheduling algorithm, data set, and CPU architecture type. The latter two are described in more detail in Section \ref{sec:data-sets} and Section \ref{sec:cpu-architectures}, respectively. If the interquartile range of the measurements corresponding to a scheduling method was too large, we rejected and re-ran all experiments on the same matrix and processor configuration.

The experiments for the schedulers HDagg and SpMP were carried out in the sympiler framework \cite{cheshmi2017sympiler, cheshmi2022transforming} as in \cite{zarebavani2022hdagg} with only minor adjustments to adhere to the aforementioned setup. All remaining schedulers were tested in our own framework. 

All scheduling algorithms are implemented in \cpp~and were compiled with GCC (11.4.0 or 11.5.0) using the optimization flag \texttt{-O3}.

\subsection{Data sets}
\label{sec:data-sets}

For the experiments, we used matrices from several data sets. The main data set is a sample from the SuiteSparse Matrix Collection \cite{davis2011university}, which constitutes a diverse set of matrices from a wide range of applications and was used in previous studies \cite{zarebavani2022hdagg}. We also consider two modified versions of this data set that are also relevant in their own right. Finally, these data sets are complemented with two randomly generated ones: uniformly random, i.e., Erd\H{o}s--R\'enyi matrices \cite{erdHos1959random}, and random with a bias 
towards the diagonal. The former are easier to parallelize as they have few (and thus large) wavefronts \cite{hasenplaugh2014ordering} and the latter are specifically designed to be harder to parallelize, though they admit good locality.

A useful general metric to understand the parallelizability of an SpTRSV execution 
is the average wavefront size, which can be calculated from the DAG representation by dividing the number of vertices by the length of the longest path. This metric is indicated for each matrix in the overview of the data sets in the supplement, see Section \ref{sec:tables-matrices}.

\subsubsection{SuiteSparse} \label{sec:florida-graphs}

From the SuiteSparse Matrix Collection \cite{davis2011university}, we used the lower triangular part of all the sparse real symmetric positive definite matrices. Out of those, we further restricted ourselves to large matrices with enough available parallelism, meaning
\begin{itemize}%[leftmargin=10pt]
	\item the number of floating point operations\footnote{The number of floating point operations is equal to twice the number of non-zeros minus the dimension of the matrix.} is at least $2$ million, and
	\item the average wavefront size is at least $44$, twice the number of cores utilized in the experiments.
\end{itemize}

We furthermore removed matrices from the data set which had the same sparsity pattern. An overview over some statistics of the matrices may be found in Table \ref{table:florida-graph-table} of the supplement.

\subsubsection{SuiteSparse METIS (METIS)} \label{sec:metis-graphs}
Zarebavani \emph{et al.\@} \cite{zarebavani2022hdagg} also use the real symmetric positive definite matrices from the SuiteSparse Matrix Collection as their data set. 
However, they use a modified version of this data set, which we also reproduce here. In their experiments, the matrices are first symmetrically permuted using a fill-reducing method of METIS \cite{karypis1998fast} and only then the lower triangular part is taken.
In general, this results in non-equivalent SpTRSV problems. The sparsity pattern of the matrices in this data set are representative of SpTRSV workloads in a Gau{\ss}--Seidel or a zero-fill-in incomplete Cholesky preconditioned conjugate gradient method for sparse symmetric solve. An overview over some statistics of the matrices may be found in Table \ref{table:hdagg-metis-graph-table} of the supplement.

\subsubsection{SuiteSparse Eigen incomplete Cholesky (iChol)} \label{sec:cholesky-graphs}

This data set consists of lower triangular matrices obtained after an incomplete Cholesky decomposition. The initial set of matrices are the same symmetric matrices used in the SuiteSparse data set\footnote{The matrix `bundle\_adj' segmentation-faults during the process and is thus excluded from the data set.}. 
The incomplete Cholesky decomposition was performed using the `IncompleteCholesky' method of Eigen \cite{eigenweb} using the built-in fill-reducing method `AMDOrdering'. An overview over some statistics of the matrices may be found in Table \ref{table:cholesky-graph-table} of the supplement.

\subsubsection{Erd\H{o}s--R\'enyi} \label{sec:Erdos-Renyi-graphs}

These are lower triangular matrices where each entry $(i,j)$, with $i > j$, is independently non-zero with a fixed probability $p$. The values of the non-zero non-diagonal entries we have chosen to be independently uniformly distributed in $[-2, 2]$. The absolute value of the diagonal entries we have chosen to be independently log-uniformly distributed in $[2^{-1}, 2]$ and their sign to be ${\pm}$ independently uniformly random\footnote{The change of distribution on the diagonal is to avoid numerical instability, in particular divisions by zero.}. The DAGs corresponding to these matrices are directed Erd\H{o}s--R\'enyi random graphs \cite{erdHos1959random}.

We generated thirty $N\times N$ matrices of this type with $N =$ 100,000 and $p = 10^{-4}, 5 \cdot 10^{-4}, 2 \cdot 10^{-3}$, ten of each given probability. An overview over some statistics of the matrices may be found in Table~\ref{table:erdos-renyi-graph-table} of the supplement.

\subsubsection{Narrow bandwidth} \label{sec:random-bandwidth-graph}
We also create a data set of random matrices which are much harder to parallelize by design, but have good locality.
Unlike the Erd\H{o}s--R\'enyi random matrices, we let the lower triangular matrix entry $(i,j)$, with $i > j$, be independently non-zero with probability $p \cdot \exp((1+j-i)/B)$, moving the non-zero entries closer to the diagonal. The entry values were chosen as in Section~\ref{sec:Erdos-Renyi-graphs}.

We generated thirty $N\times N$ matrices of this type with $N =$ 100,000 and $(p,B) = (0.14, 10)$, $(0.05, 20)$, $(0.03, 42)$, ten for each pair $(p,B)$. An overview over some statistics of the matrices may be found in Table \ref{table:random-bandwidth-graph-table} of the supplement.

\subsection{CPU architectures}
\label{sec:cpu-architectures}

The CPU architectures used for the experiments were x86 and ARM. The precise model and some specifications are given, respectively, as follows:

\begin{itemize}
	\item Intel Xeon Gold 6238T processor (x86), with 192 GB memory and theoretical peak memory throughput of 140.8 GB/s and 22 cores on a single socket; kernel version 5.14.0; GCC version 11.5.0; 
	\item AMD EPYC 7763 processor (x86), with 1024 GB memory and theoretical peak memory throughput of 204.8 GB/s and 64 cores on a single socket; kernel version 5.15.0; GCC version 11.4.0;
	\item Huawei Kunpeng 920-4826 (Hi1620) processor (ARM), with 512 GB memory and theoretical peak memory throughput of 187.7 GB/s and 48 cores on a single socket; kernel version 5.15.0; GCC version 11.4.0.
\end{itemize}

\section{Evaluation}
\label{sec:evaluation}

\subsection{Overall performance}

We present speed-ups of the forward-/backward-substitution algorithm based on parallel schedules compared to serial execution. 
The schedules of our proposed algorithm are benchmarked against those produced by the baseline methods, 
SpMP \cite{park2014sparsifying} and HDagg \cite{zarebavani2022hdagg}. 
The results, aggregated over the instances from the respective data set using the geometric mean of all pairs of runs, 
are displayed in Table~\ref{table:main-speed-up}. All experiments were conducted on the Intel x86 machine utilizing 22 cores.

On our main data set, SuiteSparse, the schedules generated by our GrowLocal algorithm achieves a geometric-mean speed-up of $1.42\times$ compared to SpMP and $3.32\times$ compared to HDagg. We also see similar results on the two variations of the SparseSuite data set: on METIS, GrowLocal obtains a $1.70\times$ and $1.77\times$ geometric-mean speed-up to SpMP and HDagg, respectively, and on iChol, it achieves a $1.80\times$ and $2.20\times$ speed-up to SpMP and HDagg, respectively. This shows that GrowLocal indeed significantly outperforms the baseline algorithms on these application-based data sets.

The differences are even larger on the Narrow Bandwidth matrices: here GrowLocal achieves a $2.50\times$ and $10.12\times$ factor improvement compared to SpMP and HDagg, respectively. This indicates that in the more challenging cases when our DAGs are particularly hard to parallelize, GrowLocal is even more superior to the baselines.

Finally, on the Erd\H{o}s--R\'enyi data set, the improvement is much smaller; since these DAGs are easier to parallelize, the differences between the algorithms become less relevant.

We also note that Funnel coarsening does not seem to further improve GrowLocal; we elaborate on this later in Section~\ref{sec:eval-funnel}.

\begin{table}[!htp]
\centering
\begin{tblr}{
		colspec = {  Q[l,m] || Q[c,m] | Q[c,m] | Q[si={table-format=1.2, table-number-alignment=center},c,m] | Q[si={table-format=1.2, table-number-alignment=center},c,m] },
		row{1} = {font=\bfseries, guard},
		rowhead = 1,
		rowfoot = 0,
	}
Data set & $\!\!$GrowLocal$\!\!$  & $\!\!$Funnel+GL$\!\!$ & SpMP & HDagg  \\ \hline \hline
SuiteSparse & \makebox[5ex][r]{\textbf{10.79}} & \makebox[5ex][r]{10.19} & 7.60 & 3.25 \\ \hline
METIS & \makebox[5ex][r]{\textbf{15.93}} & \makebox[5ex][r]{15.40} & 9.35 & 9.00 \\ \hline
iChol & \makebox[5ex][r]{\textbf{15.10}} & \makebox[5ex][r]{14.84} & 8.36 & 6.87 \\ \hline
Erd\H{o}s--R\'enyi	& \makebox[5ex][r]{\textbf{12.75}} & \makebox[5ex][r]{12.66} & 9.38 & 8.44  \\ \hline
Narr. bandw. & \makebox[5ex][r]{\textbf{9.04}} & \makebox[5ex][r]{8.26} & 3.56 & 0.88 
\end{tblr}
\medskip
\caption{Geometric mean of speed-ups over serial execution of GrowLocal with/without Funnel coarsening, 
compared to the baselines SpMP and HDagg on the Intel x86 machine using 22 cores taken over the data sets from Section \ref{sec:data-sets}.} \label{table:main-speed-up}
\end{table}

We also include a performance profile \cite{dolan2002benchmarking} based on the data generated from the SuiteSparse data set in Figure \ref{fig:performance_plot_florida}. The closer the line is to the top left corner, the better and more consistent the algorithm is across the data set. This shows that our algorithm is not only faster in execution time on average but it is so throughout the diverse SuiteSparse data set.

\begin{figure}[!htp]
	\centering
	\includegraphics[scale=0.45]{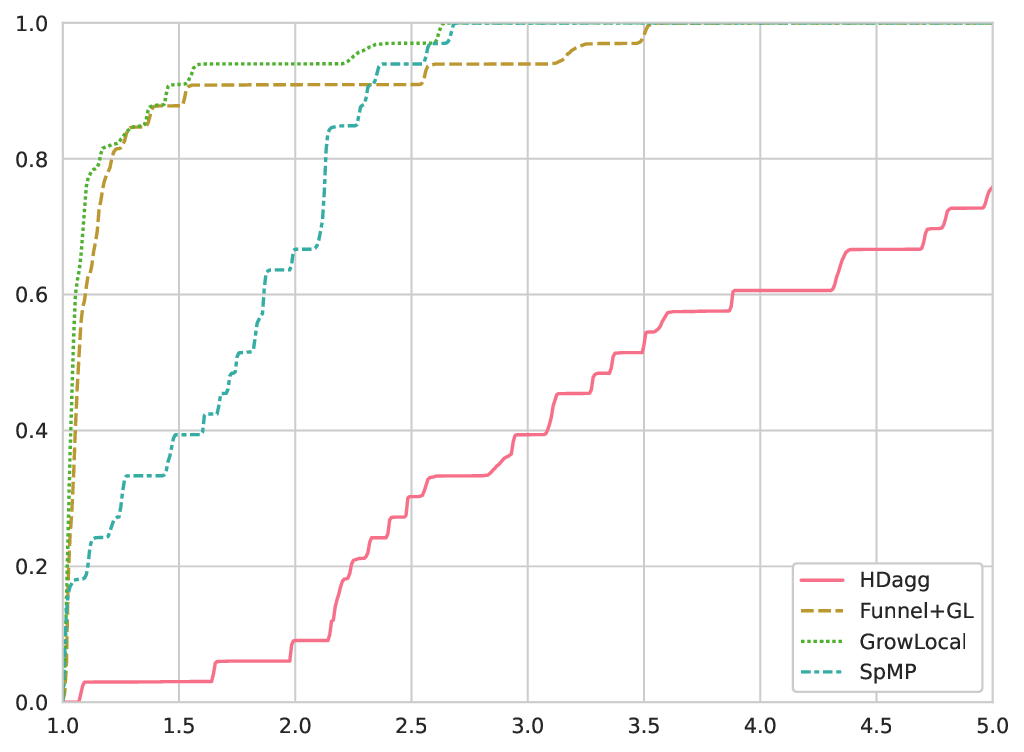}
	\captionof{figure}{Performance profiles of our algorithms on the SuiteSparse data set evaluated on the Intel x86 machine using 22 cores. The x-axis represents a threshold and the y-axis is the proportion of runs that are within this threshold times the fastest SpTRSV run on the respective matrix.}
	\label{fig:performance_plot_florida}
\end{figure}

\subsection{Fewer synchronization barriers}

The results in Table \ref{table:main-speed-up} show that our scheduler can significantly outperform the synchronous state-of-the-art HDagg. A further analysis shows that this is in part due to a substantial reduction in the number of synchronization barriers required during execution, whilst still maintaining a good work balance. In particular, Table~\ref{table:reduction-synch-barriers} shows the number of synchronization barriers relative to the number of wavefronts in our algorithm and HDagg. The data indicates a large, up to $14.99\times$, reduction of number of synchronization barriers compared to the number of wavefronts on the SuiteSparse data set. This is a reduction of up to $12.07\times$ compared to HDagg, which explains the significant speed-ups achieved by our methods. In general, we can observe a similar effect in the remaining data sets, but the difference is much smaller for Erd\H{o}s--R\'enyi, and much higher for narrow bandwidth matrices.

\begin{table}[!htp]
	\centering
	\begin{tblr}{
			colspec = {  Q[l,m] || Q[c,m] | Q[c,m] | Q[c,m] },
			row{1} = {font=\bfseries, guard},
			rowhead = 1,
			rowfoot = 0,
		}
		Data set & GrowLocal  & Funnel+GL & HDagg  \\ \hline \hline
		SuiteSparse & \makebox[5ex][r]{14.99} & \makebox[5ex][r]{\textbf{17.09}} & 1.24 \\ \hline
		METIS & \makebox[5ex][r]{16.55} & \makebox[5ex][r]{\textbf{21.83}} & 2.39 \\ \hline
		iChol & \makebox[5ex][r]{18.91} & \makebox[5ex][r]{\textbf{22.86}} & 1.62 \\ \hline
		Erd\H{o}s--R\'enyi & \makebox[5ex][r]{2.93} & \makebox[5ex][r]{\textbf{2.99}} &  1.25 \\ \hline
		Narrow bandw. & \makebox[5ex][r]{\textbf{51.12}} & \makebox[5ex][r]{42.00} & 1.10
	\end{tblr}
	\medskip
	\caption{Geometric mean of the reduction of the number of synchronization barriers relative to the number of wavefronts of the matrix within each data set from Section \ref{sec:data-sets}.} \label{table:reduction-synch-barriers}
\end{table}

\subsection{Impact of Funnel coarsening} \label{sec:eval-funnel}

The results in Table~\ref{table:main-speed-up} show that the DAG coarsening approach does not allow to further improve the schedules developed by GrowLocal on most of the data sets. This is an interesting contrast to the HDagg baseline, where coarsening is also a key ingredient of the scheduler. This suggests that GrowLocal is already rather strong at exploring the DAG structure and exploiting locality, and hence, the advantages of the coarsening step cannot compensate for the loss of structure in the graph.

However, besides this negative result, the Funnel coarsener also has several benefits that make it interesting in its own right. Firstly, it allows to run GrowLocal on a much smaller DAG, and as a results, the combined running time of Funnel+GrowLocal is often lower than GrowLocal alone; we will quantify this later in Section~\ref{sec:profitability}. As such, Funnel+GrowLocal can be a more desirable alternative when the scheduling time is also a critical factor. Secondly, Funnel coarsening allows one to reduce the number of synchronization barriers even further: while GrowLocal achieves a $12.07\times$ geo-mean reduction compared to HDagg, Funnel+GrowLocal together achieves a $13.76\times$ geo-mean reduction of synchronization barriers, which is of independent interest.

\subsection{Impact of reordering}

We separately analyze the impact of the reordering step on the performance. Table \ref{table:reorder-comparison} compares the speed-ups achieved by our algorithm with and without the reordering component from Section~\ref{sec:reorder}. The numbers show that reordering is indeed a valuable ingredient of our scheduler. The data also confirms that even without the reordering, the algorithm still outperforms HDagg notably, which is the current state-of-the-art synchronous baseline, cf.\@ Table \ref{table:main-speed-up}. 

\begin{table}[!htp]
	\centering
	\begin{tblr}{
			colspec = {  Q[l,m] || Q[si={table-format=2.2, table-number-alignment=center},c,m] | Q[si={table-format=2.2, table-number-alignment=center},c,m] },
			row{1} = {font=\bfseries, guard},
			rowhead = 1,
			rowfoot = 0,
		}
		Data set & Reordering  & No Reordering  \\ \hline \hline
		SuiteSparse & 10.79 & 8.62 \\ \hline
        METIS & 15.93 & 15.21  \\ \hline
        iChol & 15.10 & 15.02 \\ \hline
		Erd\H{o}s--R\'enyi	& 12.75 & 7.87  \\ \hline
		Narrow bandw. & 9.04 & 6.96 
	\end{tblr}
\medskip
	\caption{Geometric mean of speed-ups relative to Serial of GrowLocal with/without permuting the matrix data according to the computed schedule. Experiments were conducted on the Intel x86 machine using 22 cores.
	} \label{table:reorder-comparison}
\end{table}

\subsection{Performance across different architectures}

We show the performance gains of our algorithm over the different processors and architectures in Table~\ref{table:arch-comparison}. The data confirms that our algorithm consistently outperforms the baselines across all considered architectures. We note that the improvement relative to Serial can be in a significantly different range due to the properties of the distinct architectures. SpMP is omitted for the ARM architecture because its implementation is x86-specific.

\begin{table}[!htpb]
	\centering
	\begin{tblr}{
			colspec = {  Q[l,m] || Q[si={table-format=2.2, table-number-alignment=center},c,m] | Q[si={table-format=1.2, table-number-alignment=center},c,m] | Q[si={table-format=1.2, table-number-alignment=center},c,m] },
			row{1} = {font=\bfseries, guard},
			rowhead = 1,
			rowfoot = 0,
			cell{4}{3} = {guard},
		}
		Machine & $\!\!$GrowLocal$\!\!$  &  SpMP & HDagg  \\ \hline \hline
		Intel x86 & 10.79 & 7.60 & 3.25  \\ \hline
		AMD x86	& 5.20 & 3.65 & 1.98  \\ \hline
		$\!\!$ Huawei ARM $\!\!$ & 9.27  & n/a & 2.16		
	\end{tblr}
\medskip
	\caption{Geometric mean speed-ups relative to Serial of our algorithms over different machines and processor architectures. Experiments were conducted using 22 cores on the SuiteSparse data set.} \label{table:arch-comparison}
\end{table}

\subsection{Scaling with the number of cores} \label{sec:scaling}

Another natural question is how our algorithm scales with a growing number of cores. To examine this, we illustrate the speed-ups (over serial execution) for different numbers of cores in Table \ref{table:nr-core-scaling}. We note that this experiment was conducted on the AMD x86 machine as it has $64$ available cores on a single socket.
\begin{table}[!htp]
	\centering
	\begin{tblr}{
			colspec = {  Q[l,m] || Q[c,m] | Q[c,m] | Q[c,m] | Q[c,m]  | Q[c,m]  | Q[c,m] },
			row{1} = {font=\bfseries},
			rowhead = 1,
			rowfoot = 0,
		}
		Algorithm   & 4 &  16 & 32 & 48 & 56 & 64  \\ \hline \hline
		GrowLocal  & 2.63 & 4.15 & 5.34 & 5.70 & 5.76 & 5.85 \\ 
	\end{tblr}
\medskip
	\caption{Geometric mean of speed-ups relative to Serial of GrowLocal  for different number of cores on the AMD x86 machine taken over the SuiteSparse data set.
	} \label{table:nr-core-scaling}
\end{table}
As one sees, additional cores have diminished or negative returns at the higher end of number of cores. A reason for this is the average wavefront size which is a proxy for the amount of parallelism available. If we split the SuiteSparse data set into groups according to their average wavefront size, we see that these groups scale to different number of cores, see Figure \ref{fig:scaling_epyc}. This shows that our algorithm does scale if the matrices allow for it. 
\begin{figure}[!htp]
\centering
\includegraphics[scale=0.54]{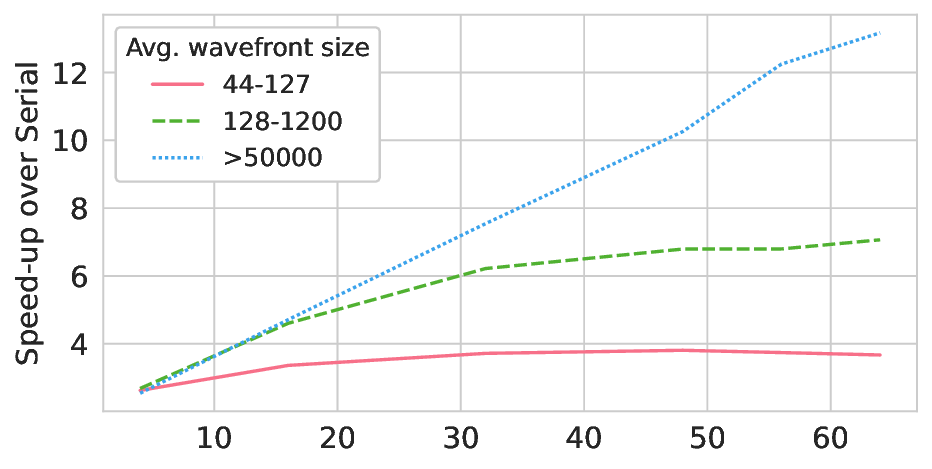}
	\caption{Geometric mean speed-ups of GrowLocal for different number of cores on the AMD x86 machine taken over the SuiteSparse data set categorized by average wavefront size.}
\label{fig:scaling_epyc}
\end{figure}

\subsection{Amortization of scheduling time}
\label{sec:profitability}

In this section, we consider the gain of the different scheduling algorithms when the scheduling time is taken into account. We measure the amortization threshold as the following ratio\footnote{If the parallel execution is slower than the serial one, then the amortization threshold is defined as $+ \infty$.}: 
\begin{equation}
	\frac{\text{scheduling\_time}}{\text{serial\_execution\_time} - \text{parallel\_execution\_time}} \: .
\end{equation}
The same metric was considered by Zarebavani \emph{et. al.} \cite[\S V.B]{zarebavani2022hdagg} and expresses how often the schedule needs to be reused in order to justify the time spent on computing it. Table \ref{table:profitability} presents the amortization threshold for GrowLocal and the two baselines, with the 25th percentile, median, and 75th percentile values shown for each algorithm.
The results indicate that the amortization threshold of GrowLocal (both with and without Funnel) is of a similar magnitude to that of SpMP. 
HDagg has a significantly higher amortization threshold on the SparseSuite data set. However, we remark that, e.g., on the METIS pre-processed data set we observed a comparable median amortization threshold of $44.21$ for HDagg.

\begin{table}[!htpb]
	\centering
	\begin{tblr}{
			colspec = {  Q[l,m] || Q[si={table-format=3.2, table-number-alignment=center},c] | Q[si={table-format=3.2, table-number-alignment=center},c] | Q[si={table-format=4.2, table-number-alignment=center},c]  },
			row{1} = {font=\bfseries, guard},
			rowhead = 1,
			rowfoot = 0,
			cell{6}{2} = {font=\bfseries},
			cell{6}{3} = {font=\bfseries},
			cell{6}{4} = {font=\bfseries},
		}
		Algorithm & Q25 & Median & Q75  \\ \hline \hline
		GrowLocal & 23.78 & 26.12 & 30.28 \\ \hline
		Funnel+GL & 17.78 & 21.74 & 27.78 \\ \hline
		SpMP	& 3.65 & 5.51  & 8.41   \\ \hline
		HDagg	& 311.23 & 961.39  & 1848.80 
		
	\end{tblr}
	\medskip
	\caption{Amortization threshold of several scheduling algorithms on the SuiteSparse data set with the 25th percentile, median, and 75th percentile values shown for each algorithm. The data was collected on the Intel x86 machine using 22 cores.
	} \label{table:profitability}
\end{table}

\subsection{Block parallel scheduling} \label{sec:eval-block-parallel scheduling}

In order to further reduce the amortization threshold of GrowLocal, we now consider the effect of subdividing the matrix into blocks and applying the GrowLocal scheduler on each block in parallel, cf.\@ Section \ref{sec:growlocal-parallel-block-schedule}. In Table \ref{table:sched-threads-effects}, we record the effect of running GrowLocal with multiple scheduling threads on scheduling time, floating point operations per second, number of supersteps, and amortization threshold. We see that using multiple scheduling threads can lead to super-linear speed-up of scheduling time. This is because there are now several (long) egdes that never have to be considered, cf.\@ Figure \ref{fig:block-decomp}. We also see that the SpTRSV solve time is more significantly affected when using a higher number of scheduling threads. We note that these effects are heavily dependent on the matrix. Some matrices, such as `af\_shell7' and `bmwcra\_1' have approximately a $30\%$ performance drop in the SpTRSV solve time already when using just two blocks, whereas `bundle\_adj' and `Hook\_1498' are hardly affected even when using $16$ blocks. Despite the performance drops, we find a near linear decrease in the amortization threshold.

Table~\ref{table:sched-threads-effects} indicates, for instance, that using $6$ scheduling threads is a noteworthy compromise. This lowers the median amortization threshold to $4.54$, which is smaller than that of SpMP, cf.\ Table~\ref{table:profitability}, whilst maintaining a $1.05\times$ speed-up over SpMP. This makes GrowLocal superior on both metrics simultaneously.

\begin{table}[htpb]
	\centering
	\begin{tblr}{
			colspec = {  Q[c,m] || Q[si={table-format=2.2, table-number-alignment=center},c,m] | Q[si={table-format=1.2, table-number-alignment=center},c,m] | Q[si={table-format=1.2, table-number-alignment=center},c,m] | Q[si={table-format=2.2, table-number-alignment=center},c,m] },
			row{1} = {font=\bfseries, guard},
			rowhead = 1,
			rowfoot = 0,
		}
		Threads & {Sched.\@ \\ Time} & Flops/s & Supersteps & {Amort.\@ \\ Threshold }   \\ \hline \hline
		1  &  1.00 & 1.00 & 1.00 & 26.12  \\ \hline
		2  &  2.01 & 0.89 & 1.47 & 13.59  \\ \hline
		4  &  4.11 & 0.79 & 1.99 &  6.91  \\ \hline
		6  &  6.28 & 0.74 & 2.35 &  4.54  \\ \hline
		8  &  8.34 & 0.70 & 2.66 &  3.48  \\ \hline
		16 & 17.06 & 0.57 & 3.84 &  1.78  \\ \hline
		22 & 23.43 & 0.52 & 4.53 &  1.31
	\end{tblr}
\medskip
	\caption{Geometric means of relative speed-up of scheduling time, relative decrease in double precision floating point operations per second, and relative increase of number of supersteps of GrowLocal compared to using just a single scheduling thread, that is, a single scheduling block, together with the median amortization threshold of GrowLocal on the SuiteSparse data set. The data was collected on the Intel x86 machine using 22 cores.} \label{table:sched-threads-effects}
\end{table}

\section{Conclusion and future directions}

The results show that our GrowLocal scheduler indeed significantly speeds up the parallel SpTRSV kernel, reducing the execution time by a $1.42\times$ geometric-mean factor compared to SpMP and $3.32\times$ compared to HDagg on the SuiteSparse benchmark. The data also shows that the algorithm performs similarly well on multiple other data sets, and the improvements are consistent over various architectures. The scheduling time of GrowLocal is also competitive with the baselines, especially when combined with the block decomposition technique.

Future work may consider the adaptation of our algorithm to non-uniform memory access (NUMA) architectures. In particular, the AMD x86 data in Section~\ref{sec:scaling} confirms that our algorithm scales well to a high number of cores. However, when
solving SpTRSV on highly NUMA architectures, we expect 
 the parallel execution to be less effective. In order to adapt to such a NUMA setting, one should consider fundamental changes to the SpTRSV kernel and the data structures, such as partitioning the matrix local to threads, interleaving the vector, and reducing the need for global synchronization. 
 It is also an interesting question whether the scheduling algorithm can be efficiently adapted to NUMA, for example, by considering non-uniform bandwidth or latency. 
 To our knowledge, there are currently no scheduling algorithms for SpTRSV that directly account for such NUMA effects.

Another promising direction for future work is to combine our barrier list scheduling algorithms with other approaches that proved successful for SpTRSV in the past. For instance, one could seek to adapt GrowLocal to a semi-asynchronous setting as in SpMP, in order to allow for a more flexible parallel execution. This could allow for further speed-ups on top of our current results.

\clearpage

\appendix

\section{Tables of matrices} \label{sec:tables-matrices}

\noindent Here we provide some basic statistics of the matrices used in the experiments, cf.\@ Section \ref{sec:data-sets}.

\begin{longtblr}[
	theme = ams-theme,
	caption = {Matrices and statistics from SuiteSparse Matrix Collection \cite{davis2011university} used for the evaluation. The average wavefront size (Avg. wf) has been rounded down.},
	label = {table:florida-graph-table},
]{
	colspec = { Q[l,m] | Q[r,m] | Q[r,m] | Q[r,m]  },
	row{1} = {font=\bfseries},
	rowhead = 1,
	rowfoot = 0,
}
Matrix & Size & \#Non-zeros & Avg.\@ wf
\\ \hline \hline af\_0\_k101 & 503,625 & 9,027,150 & 74
\\ \hline af\_shell7 & 504,855 & 9,046,865 & 135
\\ \hline apache2 & 715,176 & 2,766,523 & 1,077
\\ \hline audikw\_1 & 943,695 & 39,297,771 & 203
\\ \hline bmw7st\_1 & 141,347 & 3,740,507 & 199
\\ \hline bmwcra\_1 & 148,770 & 5,396,386 & 204
\\ \hline bone010 & 986,703 & 36,326,514 & 470
\\ \hline boneS01 & 127,224 & 3,421,188 & 156
\\ \hline boneS10 & 914,898 & 28,191,660 & 386
\\ \hline Bump\_2911 & 2,911,419 & 65,320,659 & 283
\\ \hline bundle\_adj & 513,351 & 10,360,701 & 57,039
\\ \hline consph & 83,334 & 3,046,907 & 139
\\ \hline Dubcova3 & 146,689 & 1,891,669 & 44
\\ \hline ecology2 & 999,999 & 2,997,995 & 500
\\ \hline Emilia\_923 & 923,136 & 20,964,171 & 176
\\ \hline Fault\_639 & 638,802 & 14,626,683 & 143
\\ \hline Flan\_1565 & 1,564,794 & 59,485,419 & 200
\\ \hline G3\_circuit & 1,585,478 & 4,623,152 & 611
\\ \hline Geo\_1438 & 1,437,960 & 32,297,325 & 246
\\ \hline hood & 220,542 & 5,494,489 & 365
\\ \hline Hook\_1498 & 1,498,023 & 31,207,734 & 95
\\ \hline inline\_1 & 503,712 & 18,660,027 & 287
\\ \hline ldoor & 952,203 & 23,737,339 & 141
\\ \hline msdoor & 415,863 & 10,328,399 & 59
\\ \hline offshore & 259,789 & 2,251,231 & 75
\\ \hline parabolic\_fem & 525,825 & 2,100,225 & 75,117
\\ \hline PFlow\_742 & 742,793 & 18,940,627 & 118
\\ \hline Queen\_4147 & 4,147,110 & 166,823,197 & 342
\\ \hline s3dkt3m2 & 90,449 & 1,921,955 & 60
\\ \hline Serena & 1,391,349 & 32,961,525 & 298
\\ \hline shipsec1 & 140,874 & 3,977,139 & 67
\\ \hline StocF-1465 & 1,465,137 & 11,235,263 & 487
\\ \hline thermal2 & 1,228,045 & 4,904,179 & 991
\end{longtblr}

\begin{longtblr}[
theme = ams-theme,
caption = {Matrices and statistics from SuiteSparse Matrix Collection \cite{davis2011university} symmetrically permuted using the fill-reducing method `METIS\textunderscore NodeND' of \cite{karypis1998fast}. The average wavefront size (Avg. wf) has been rounded down.},
label = {table:hdagg-metis-graph-table},
]{ colspec = { Q[l,m] | Q[r,m] | Q[r,m] | Q[r,m]  }, row{1} = {font=\bfseries}, rowhead = 1, rowfoot = 0}
Matrix & Size & \#Non-zeros & Avg. wf
\\ \hline \hline af\_0\_k101\_metis & 503,625 & 9,027,150 & 610
\\ \hline af\_shell10\_metis & 1,508,065 & 27,090,195 & 1,065
%\\ \hline af\_shell7\_metis & 504,855 & 9,046,865 & 917
\\ \hline apache2\_metis & 715,176 & 2,766,523 & 47,678
\\ \hline audikw\_1\_metis & 943,695 & 39,297,771 & 1,734
%\\ \hline bmw7st\_1\_metis & 141,347 & 3,740,507 & 654
\\ \hline bmwcra\_1\_metis & 148,770 & 5,396,386 & 473
\\ \hline bone010\_metis & 986,703 & 36,326,514 & 1,326
%\\ \hline boneS01\_metis & 127,224 & 3,421,188 & 517
\\ \hline boneS10\_metis & 914,898 & 28,191,660 & 2,401
%\\ \hline Bump\_2911\_metis & 2,911,419 & 65,320,659 & 4,901
\\ \hline bundle\_adj\_metis & 513,351 & 10,360,701 & 11,407
\\ \hline cant\_metis & 62,451 & 2,034,917 & 333
\\ \hline consph\_metis & 83,334 & 3,046,907 & 247
\\ \hline crankseg\_2\_metis & 63,838 & 7,106,348 & 86
%\\ \hline Dubcova3\_metis & 146,689 & 1,891,669 & 1,725
\\ \hline ecology2\_metis & 999,999 & 2,997,995 & 62,499
\\ \hline Emilia\_923\_metis & 923,136 & 20,964,171 & 2,107
\\ \hline Fault\_639\_metis & 638,802 & 14,626,683 & 1,458
\\ \hline Flan\_1565\_metis & 1,564,794 & 59,485,419 & 2,569
\\ \hline G3\_circuit\_metis & 1,585,478 & 4,623,152 & 93,263
\\ \hline Geo\_1438\_metis & 1,437,960 & 32,297,325 & 2,887
\\ \hline gyro\_metis & 17,361 & 519,260 & 88
\\ \hline hood\_metis & 220,542 & 5,494,489 & 984
\\ \hline Hook\_1498\_metis & 1,498,023 & 31,207,734 & 4,059
\\ \hline inline\_1\_metis & 503,712 & 18,660,027 & 1,549
\\ \hline ldoor\_metis & 952,203 & 23,737,339 & 4,858
\\ \hline m\_t1\_metis & 97,578 & 4,925,574 & 268
\\ \hline msdoor\_metis & 415,863 & 10,328,399 & 1,856
\\ \hline nasasrb\_metis & 54,870 & 1,366,097 & 287
%\\ \hline offshore\_metis & 259,789 & 2,251,231 & 3,711
%\\ \hline oilpan\_metis & 73,752 & 1,835,470 & 277
%\\ \hline parabolic\_fem\_metis & 525,825 & 2,100,225 & 17,527
\\ \hline PFlow\_742\_metis & 742,793 & 18,940,627 & 1,023
\\ \hline pwtk\_metis & 217,918 & 5,926,171 & 511
%\\ \hline Queen\_4147\_metis & 4,147,110 & 166,823,197 & 3,736
\\ \hline raefsky4\_metis & 19,779 & 674,195 & 111
%\\ \hline s3dkt3m2\_metis & 90,449 & 1,921,955 & 162
%\\ \hline Serena\_metis & 1,391,349 & 32,961,525 & 3,540
\\ \hline ship\_003\_metis & 121,728 & 4,103,881 & 494
%\\ \hline shipsec1\_metis & 140,874 & 3,977,139 & 479
\\ \hline shipsec8\_metis & 114,919 & 3,384,159 & 456
\\ \hline StocF-1465\_metis & 1,465,137 & 11,235,263 & 11,446
\\ \hline thermal2\_metis & 1,228,045 & 4,904,179 & 45,483
\\ \hline tmt\_sym\_metis & 726,713 & 2,903,837 & 26,915
\\ \hline x104\_metis & 108,384 & 5,138,004 & 306
\end{longtblr}

\begin{longtblr}[
theme = ams-theme,
caption = {Matrices and statistics from SuiteSparse Matrix Collection \cite{davis2011university} post Eigen incomplete Cholesky \cite{eigenweb} used for the evaluation. The average wavefront size (Avg. wf) has been rounded down.},
label = {table:cholesky-graph-table},
]{ colspec = { Q[l,m] | Q[r,m] | Q[r,m] | Q[r,m]  }, row{1} = {font=\bfseries}, rowhead = 1, rowfoot = 0}
Matrix & Size & \#Non-zeros & Avg. wf
\\ \hline \hline af\_0\_k101\_iCh & 503,625 & 9,027,150 & 195
\\ \hline af\_shell7\_iCh & 504,855 & 9,046,865 & 668
\\ \hline apache2\_iCh & 715,176 & 2,766,523 & 79,464
\\ \hline audikw\_1\_iCh & 943,695 & 39,297,771 & 138
\\ \hline bmw7st\_1\_iCh & 141,347 & 3,740,507 & 340
\\ \hline bmwcra\_1\_iCh & 148,770 & 5,396,386 & 89
\\ \hline bone010\_iCh & 986,703 & 36,326,514 & 340
\\ \hline boneS01\_iCh & 127,224 & 3,421,188 & 245
\\ \hline boneS10\_iCh & 914,898 & 28,191,660 & 521
\\ \hline Bump\_2911\_iCh & 2,911,419 & 65,320,659 & 1,048
\\ \hline consph\_iCh & 83,334 & 3,046,907 & 78
\\ \hline Dubcova3\_iCh & 146,689 & 1,891,669 & 1,594
\\ \hline ecology2\_iCh & 999,999 & 2,997,995 & 142,857
\\ \hline Emilia\_923\_iCh & 923,136 & 20,964,171 & 511
\\ \hline Fault\_639\_iCh & 638,802 & 14,626,683 & 422
\\ \hline Flan\_1565\_iCh & 1,564,794 & 59,485,419 & 689
\\ \hline G3\_circuit\_iCh & 1,585,478 & 4,623,152 & 88,082
\\ \hline Geo\_1438\_iCh & 1,437,960 & 32,297,325 & 768
\\ \hline hood\_iCh & 220,542 & 5,494,489 & 1,050
\\ \hline Hook\_1498\_iCh & 1,498,023 & 31,207,734 & 649
\\ \hline inline\_1\_iCh & 503,712 & 18,660,027 & 679
\\ \hline ldoor\_iCh & 952,203 & 23,737,339 & 3,317
\\ \hline msdoor\_iCh & 415,863 & 10,328,399 & 956
\\ \hline offshore\_iCh & 259,789 & 2,251,231 & 1,114
\\ \hline parabolic\_fem\_iCh & 525,825 & 2,100,225 & 19,475
\\ \hline PFlow\_742\_iCh & 742,793 & 18,940,627 & 240
\\ \hline Queen\_4147\_iCh & 4,147,110 & 166,823,197 & 719
\\ \hline s3dkt3m2\_iCh & 90,449 & 1,921,955 & 104
\\ \hline Serena\_iCh & 1,391,349 & 32,961,525 & 940
\\ \hline shipsec1\_iCh & 140,874 & 3,977,139 & 259
\\ \hline StocF-1465\_iCh & 1,465,137 & 11,235,263 & 2,990
\\ \hline thermal2\_iCh & 1,228,045 & 4,904,179 & 47,232
\end{longtblr}

\begin{longtblr}[
theme = ams-theme,
caption = {Matrices and statistics in the Erd\H{o}s--R\'enyi data set used for the evaluation. The average wavefront size (Avg. wf) has been rounded down.},
label = {table:erdos-renyi-graph-table},
]{ colspec = { l | r | r | r  }, row{1} = {font=\bfseries}, rowhead = 1, rowfoot = 0}
Matrix & Size & \#Non-zeroes & Avg.\@ wf
\\ \hline \hline ER\_100k\_19m\_A & 100,000 & 19,999,021 & 109
\\ \hline ER\_100k\_19m\_B & 100,000 & 19,998,182 & 109
\\ \hline ER\_100k\_19m\_C & 100,000 & 19,997,897 & 107
\\ \hline ER\_100k\_19m\_D & 100,000 & 19,995,405 & 106
\\ \hline ER\_100k\_19m\_E & 100,000 & 19,994,516 & 107
\\ \hline ER\_100k\_19m\_G & 100,000 & 19,989,535 & 106
\\ \hline ER\_100k\_19m\_H & 100,000 & 19,999,989 & 110
\\ \hline ER\_100k\_1m\_A & 100,000 & 1,001,528 & 1,785
\\ \hline ER\_100k\_1m\_B & 100,000 & 1,000,452 & 1,818
\\ \hline ER\_100k\_1m\_C & 100,000 & 1,000,315 & 1,818
\\ \hline ER\_100k\_1m\_E & 100,000 & 1,000,044 & 1,666
\\ \hline ER\_100k\_1m\_F & 100,000 & 1,000,406 & 1,785
\\ \hline ER\_100k\_1m\_G & 100,000 & 1,001,171 & 1,724
\\ \hline ER\_100k\_1m\_H & 100,000 & 1,001,551 & 1,886
\\ \hline ER\_100k\_1m\_I & 100,000 & 1,000,237 & 1,639
\\ \hline ER\_100k\_1m\_J & 100,000 & 1,001,533 & 1,851
\\ \hline ER\_100k\_20m\_F & 100,000 & 20,001,732 & 107
\\ \hline ER\_100k\_20m\_I & 100,000 & 20,006,442 & 109
\\ \hline ER\_100k\_20m\_J & 100,000 & 20,003,479 & 109
\\ \hline ER\_100k\_4m\_A & 100,000 & 4,998,205 & 395
\\ \hline ER\_100k\_4m\_C & 100,000 & 4,999,271 & 398
\\ \hline ER\_100k\_4m\_G & 100,000 & 4,999,358 & 401
\\ \hline ER\_100k\_4m\_J & 100,000 & 4,996,501 & 414
\\ \hline ER\_100k\_5m\_B & 100,000 & 5,006,107 & 411
\\ \hline ER\_100k\_5m\_D & 100,000 & 5,001,575 & 404
\\ \hline ER\_100k\_5m\_E & 100,000 & 5,004,251 & 400
\\ \hline ER\_100k\_5m\_F & 100,000 & 5,002,190 & 400
\\ \hline ER\_100k\_5m\_H & 100,000 & 5,000,573 & 409
\\ \hline ER\_100k\_5m\_I & 100,000 & 5,001,846 & 400
\\ \hline ER\_100k\_999k\_D & 100,000 & 999,915 & 1,818
\end{longtblr}

\begin{longtblr}[
theme = ams-theme,
caption = {Matrices and statistics in the narrow bandwidth data set used for the evaluation. The average wavefront size (Avg. wf) has been rounded down.},
label = {table:random-bandwidth-graph-table},
]{ colspec = { l | r | r | r  }, row{1} = {font=\bfseries}, rowhead = 1, rowfoot = 0}
Matrix & Size & \#Non-zeroes & Avg.\@ wf
\\ \hline \hline NB\_p14\_b10\_100k\_A & 100,000 & 146,565 & 87
\\ \hline NB\_p14\_b10\_100k\_B & 100,000 & 146,328 & 115
\\ \hline NB\_p14\_b10\_100k\_C & 100,000 & 147,201 & 61
\\ \hline NB\_p14\_b10\_100k\_D & 100,000 & 146,972 & 73
\\ \hline NB\_p14\_b10\_100k\_E & 100,000 & 147,369 & 73
\\ \hline NB\_p14\_b10\_100k\_F & 100,000 & 146,855 & 111
\\ \hline NB\_p14\_b10\_100k\_G & 100,000 & 147,350 & 132
\\ \hline NB\_p14\_b10\_100k\_H & 100,000 & 147,412 & 85
\\ \hline NB\_p14\_b10\_100k\_I & 100,000 & 147,132 & 132
\\ \hline NB\_p14\_b10\_100k\_J & 100,000 & 146,781 & 105
\\ \hline NB\_p3\_b42\_100k\_A & 100,000 & 127,045 & 46
\\ \hline NB\_p3\_b42\_100k\_B & 100,000 & 127,019 & 55
\\ \hline NB\_p3\_b42\_100k\_C & 100,000 & 127,708 & 29
\\ \hline NB\_p3\_b42\_100k\_D & 100,000 & 127,341 & 45
\\ \hline NB\_p3\_b42\_100k\_E & 100,000 & 127,569 & 67
\\ \hline NB\_p3\_b42\_100k\_F & 100,000 & 127,137 & 47
\\ \hline NB\_p3\_b42\_100k\_G & 100,000 & 127,774 & 52
\\ \hline NB\_p3\_b42\_100k\_H & 100,000 & 127,029 & 46
\\ \hline NB\_p3\_b42\_100k\_I & 100,000 & 127,475 & 39
\\ \hline NB\_p3\_b42\_100k\_J & 100,000 & 127,275 & 62
\\ \hline NB\_p5\_b20\_100k\_A & 100,000 & 102,053 & 1,298
\\ \hline NB\_p5\_b20\_100k\_B & 100,000 & 102,621 & 1,063
\\ \hline NB\_p5\_b20\_100k\_C & 100,000 & 102,021 & 1,298
\\ \hline NB\_p5\_b20\_100k\_D & 100,000 & 102,968 & 1,075
\\ \hline NB\_p5\_b20\_100k\_E & 100,000 & 102,650 & 952
\\ \hline NB\_p5\_b20\_100k\_F & 100,000 & 102,309 & 1,162
\\ \hline NB\_p5\_b20\_100k\_G & 100,000 & 103,152 & 892
\\ \hline NB\_p5\_b20\_100k\_H & 100,000 & 102,324 & 1,190
\\ \hline NB\_p5\_b20\_100k\_I & 100,000 & 102,465 & 1,369
\\ \hline NB\_p5\_b20\_100k\_J & 100,000 & 102,244 & 1,010
\end{longtblr}

\clearpage

\section{Time and space complexity of GrowLocal} \label{sec:growlocal-complexity}

Below we provide a more detailed discussion of the time and space complexity of GrowLocal. We first restate Theorem \ref{thm:GrowLocal-complexity} more formally, including the assumption that the out-degrees and compute weights are within a constant factor. 

\begin{restatethm}{\ref*{thm:GrowLocal-complexity}}[formal]
Assume that there exist positive constants $\eta$ and $\varrho$ such that for all vertices $u, v \in V$, we have 
\begin{align}
\omega(u) &\le \eta \cdot \omega(v), \label{eq:near-homogenous-weight} \\
\deg^+(u) &\le \varrho \cdot |E| / |V|. \label{eq:near-homogenous-out-degree}
\end{align}
Then, the time complexity of GrowLocal is $O(|E| \cdot \log |V|)$ and the space complexity is $O(|E|)$.
\end{restatethm}

\begin{proof}
For each iteration $i$, denote by $\mathcal{V}_p^{(i)}$ the vertices which have been assigned to core $p$ in iteration $i$, by $\Gamma_p^{(i)} = |\mathcal{V}_p^{(i)}|$ the number of vertices assigned to core $p$ in iteration $i$, and by $\Omega_p^{(i)}$ the sum of weights of these vertices. From the algorithm design, specifically Line \ref{alg-line:GrowLocal-weight-balance}, we have for each iteration $i$ and each core $p$ that
\begin{equation} \label{eq:app-weight-balance-cores}
	\Omega_p^{(i)} \le \mu \Omega_1^{(i)},
\end{equation}
for some fixed positive constant $\mu$. We furthermore use the notations
\begin{align}
	\Gamma_{\max}^{(i)} &= \max_p \Gamma_p^{(i)}, \\
	\Gamma_{\Sigma}^{(i)} &= \sum_p \Gamma_p^{(i)}, \quad \text{and} \\
	\omega_{\min} &= \min_{v \in V} \omega(v)
\end{align}
for simplicity.

The key observation of the analysis is that during the formation of a superstep, the total number of vertices assigned over all the iterations is only a linear factor away from the number of vertices assigned in the final superstep. If we restrict ourselves to the first core, this is easy to see intuitively: the first iteration assigns $20$ vertices to core, the next iteration assigns $20 \cdot \tfrac{3}{2}$, the following assigns $20  \cdot  (\tfrac{3}{2})^2$, and so on. If the iteration that is accepted in the end assigns $\alpha^*$ vertices to the first core, then the preceding iterations assign at most $\alpha^*  \sum_{i=1}^\infty (\frac{2}{3})^i$ altogether, and the last examined iteration (which is rejected) possibly also assigns $\frac{3}{2} \alpha^*$; this is altogether still in $O(\alpha^*)$. Note that the ratio between the last two iterations can also be less than $\frac{3}{2}$, but this does not affect the claim.

Extending the argument above to all the cores is slightly more technical due to the different vertex weights. For a core $p$, we have
\begin{multline}
\Gamma_p^{(i)} = \sum_{v \in \mathcal{V}_p^{(i)}} 1 \le \sum_{v \in \mathcal{V}_p^{(i)}} \eta \frac{\omega(v)}{\sum_{u \in \mathcal{V}_1^{(i)}} \omega(u) \slash \sum_{u \in \mathcal{V}_1^{(i)}} 1} \\
= \eta \frac{\Omega_p^{(i)}}{\Omega_1^{(i)}} \sum_{u \in \mathcal{V}_1^{(i)}} 1  \le \eta \mu \sum_{u \in \mathcal{V}_1^{(i)}} 1 = \eta\mu \Gamma_1^{(i)}.
\end{multline}
Here, we used \eqref{eq:near-homogenous-weight} and \eqref{eq:app-weight-balance-cores}. Thus, if $\alpha^{(i)}=20 \cdot (\tfrac{3}{2})^{i-1}$ is the parameter used for iteration $i$, then we have
\begin{equation}
	\alpha^{(i)} \le \Gamma_{\max}^{(i)}  \le \eta \mu \cdot \alpha^{(i)}.
\end{equation}
This means that for iterations $i$ and $j$ with $i<j$, we get that the ratio $\Gamma_{\max}^{(i)} / \Gamma_{\max}^{(j)}$ is at most $\eta \mu \cdot (\frac{2}{3})^{j-i}$.

For the proof below, we actually require a similar upper bound on the ratio
\begin{equation} \label{eq:ratio-with-L}
	\frac{\Gamma_{\max}^{(i)} + L}{\Gamma_{\max}^{(j)} + L}
\end{equation}
instead. For this, we separate two cases, namely the first few iterations and all the remaining ones. Recall that $L$ was chosen as a constant. Assume first that we have $\Gamma_{\max}^{(i)} \geq L $. In this case, we can upper bound the expression above by
\begin{equation}
	 \frac{2 \cdot \Gamma_{\max}^{(i)}}{\Gamma_{\max}^{(j)}} \leq 2\eta\mu \cdot \left( \tfrac{2}{3} \right)^{j-i} \, .
\end{equation}
On the other hand, assume that $\Gamma_{\max}^{(i)} < L $. In this case, we can upper bound (\ref{eq:ratio-with-L}) simply by $1$. The key observation is that the number of these iterations with $\Gamma_{\max}^{(i)} < L $ is at most a constant in any superstep. Indeed, starting with $\alpha=20$ and multiplying by $\frac{3}{2}$ each round, we already have $\alpha > L$ by the $\lceil \log(L) / \log(1.5) \rceil$-th iteration, and hence $\Gamma_{\max}^{(i)} \geq L$ for $i \ge C_L \coloneqq \lceil \log(L) / \log(1.5) \rceil$. As such, there are only $C_L=O(\log(L))=O(1)$ distinct values where we will use this upper bound of $1$.

Our algorithm only accepts and saves an iteration if its parallelization rate $\beta$ is at least a $0.97$ factor of the best parallelization rate observed so far during this superstep. Hence, if iteration $j$ is accepted and $i$ is any iteration such that $i<j$, then we have
\begin{equation} \label{eq:app-growlocal-score-ineq}
\frac{\sum_p \Omega_p^{(j)}}{\max_p \Omega_p^{(j)} + L} \geq 0.97 \cdot \frac{\sum_p \Omega_p^{(i)}}{\max_p \Omega_p^{(i)} + L} \, .
\end{equation}
Due to our assumption on the vertex weights, we have for any iteration $i$ that
\begin{equation}
\omega_{\min} \cdot \Gamma_{\max}^{(i)} \leq \max_p \Omega_p^{(j)} \leq \eta\omega_{\min} \cdot \Gamma_{\max}^{(i)}
\end{equation}
as well as
\begin{equation}
\omega_{\min} \cdot \Gamma_{\Sigma}^{(i)} \leq \sum_p \Omega_p^{(i)} \leq \eta\omega_{\min} \cdot \Gamma_{\Sigma}^{(i)}.
\end{equation}
Using these in Inequality \eqref{eq:app-growlocal-score-ineq}, we get that
\begin{equation}
\frac{\eta\omega_{\min} \cdot \Gamma_{\Sigma}^{(j)}}{\omega_{\min} \cdot \Gamma_{\max}^{(j)} + L} \geq 0.97 \cdot \frac{\omega_{\min} \cdot \Gamma_{\Sigma}^{(i)}}{\eta\omega_{\min} \cdot \Gamma_{\max}^{(i)} + L} \, ,
\end{equation}
which further implies
\begin{equation}
\Gamma_{\Sigma}^{(i)} \leq \frac{\eta^2}{0.97} \cdot \frac{\Gamma_{\max}^{(i)} + L}{\Gamma_{\max}^{(j)} + L} \cdot \Gamma_{\Sigma}^{(j)} \, .
\end{equation}
Using our upper bounds on (\ref{eq:ratio-with-L}) and $\eta=O(1)$, this implies
\begin{equation}
	\Gamma_{\Sigma}^{(i)} \leq O(1) \cdot \Gamma_{\Sigma}^{(j)} \, 
\end{equation}
for $i \in \{1, ..., C_L - 1\}$, and
\begin{equation}
	\Gamma_{\Sigma}^{(i)} \leq O(1) \cdot \left( \tfrac{2}{3} \right)^{j-i} \cdot \Gamma_{\Sigma}^{(j)} \, 
\end{equation}
for $i \geq C_L$. Assuming that iteration $j$ is the final worthy iteration that is accepted for our superstep, this means that the total number of assigned vertices made in iterations $i \in \{ 1, 2, ..., j-1 \}$ is at most
\begin{equation}
\sum_{i=1}^{j-1} \Gamma_{\Sigma}^{(i)} \leq O(1) \cdot \Gamma_{\Sigma}^{(j)} \cdot \left( O(1) + \sum_{\ell=1}^{j-1} \left( \tfrac{2}{3} \right)^{\ell} \right) \, .
\end{equation}
With the geometric sum upper bounded by $2$, we get that the number of assigned vertices is indeed in $O(\Gamma_{\Sigma}^{(j)})$.

Note that the argument above does not consider the possible last iteration $(j+1)$ which is rejected by our algorithm. Nevertheless, we can bound the assignments here with a similar argument. If the iteration was rejected, then its parallelization score is at most as high as that of iteration $j$, i.e., $\beta^{(j+1)} \leq \beta^{(j)}$. As before, this implies
\begin{equation}
\frac{\omega_{\min} \cdot \Gamma_{\Sigma}^{(j+1)}}{\omega_{\min} \cdot \eta \cdot \Gamma_{\max}^{(j+1)} + L} \leq \frac{\omega_{\min} \cdot \eta \cdot \Gamma_{\Sigma}^{(j)}}{\omega_{\min} \cdot \Gamma_{\max}^{(j)} + L} \, ,
\end{equation}
and hence
\begin{equation}
	\Gamma_{\Sigma}^{(j+1)} \leq O(1) \cdot \frac{\Gamma_{\max}^{(j+1)} + L}{\Gamma_{\max}^{(j)} + L} \cdot \Gamma_{\Sigma}^{(j)} \, .
\end{equation}
As
\begin{equation} \begin{aligned}
	\frac{\Gamma_{\max}^{(j+1)} + L}{\Gamma_{\max}^{(j)} + L} &=  \frac{\Gamma_{\max}^{(j+1)} - \Gamma_{\max}^{(j)}}{\Gamma_{\max}^{(j)} + L} + 1  \\ 
	&\leq \frac{\Gamma_{\max}^{(j+1)}}{\Gamma_{\max}^{(j)}} + 1  \leq \tfrac{3}{2} \eta \mu + 1 \, ,
\end{aligned}
\end{equation}
we have that $\Gamma_{\Sigma}^{(j+1)}$ is again in $O(\Gamma_{\Sigma}^{(j)})$.

As such, the number of assignments in any superstep is linear in the size of the vertices that are finally scheduled. Summing this up over all the supersteps, we get that the algorithm altogether only makes $O(|V|)$ assignments over all supersteps and iterations.

For each assignment, the chosen vertex is selected from a priority queue data structure. Each such data structure contains at most $|V|$ vertices, so the cost of each assignment is $O(\log|V|)$. However, after each assignment of a concrete vertex $v$, we also need to examine all the children $u$ of $v$, check if $u$ also becomes ready with this assignment (i.e., all parents of $u$ are computed now), and if so, then also insert $u$ into such a priority queue at a time cost of $O(\log|V|)$. Since any vertex $v$ has at most $\varrho \cdot |E| / |V|$ children and $\varrho \in O(1)$, this sums up to a total of $O(|V| \cdot |E| / |V| \cdot \log|V|)$ over all the $O(|V|)$ assignments. 
This results in an overall time complexity of $O(|E| \cdot \log|V|)$ for the algorithm.

The space complexity of the algorithm is much easier to settle: each iteration only stores $O(|V|)$ data, and we store at most two iterations at a time, so the main bottleneck here is simply storing the input DAG itself, which requires $O(|E|)$ space.
\end{proof}

For the sake of completeness, we complement the theoretical bound in Theorem \ref{thm:GrowLocal-complexity} with empirical data in Figure \ref{fig:scheduling-time-empirical}.

\begin{figure}[!htpb]
	\centering
	\includegraphics[scale=0.473]{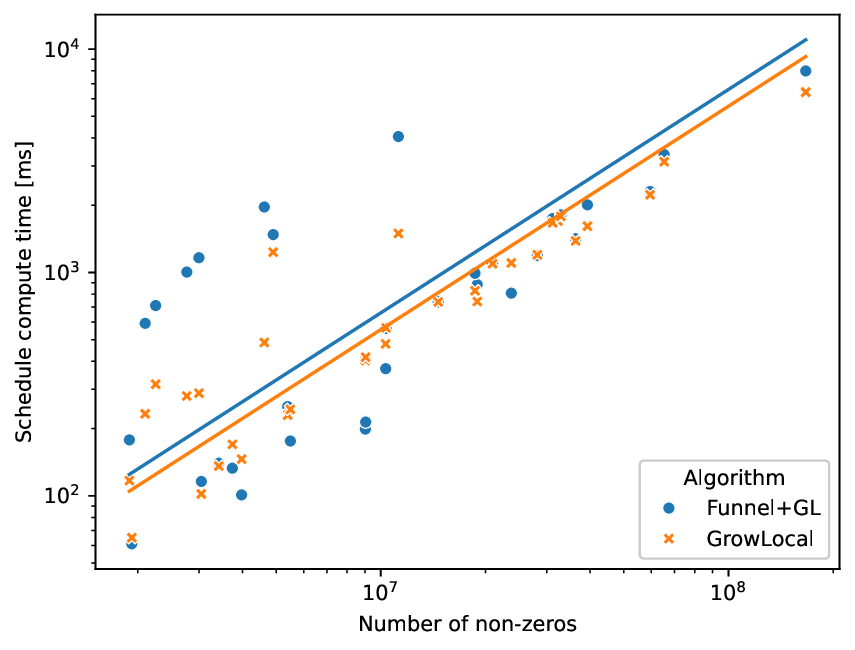}
	\caption{Scheduling time of Funnel+GL and GrowLocal on the SuiteSparse data set. The straight lines are the best square-mean-error fit of the family of curves $\log(y) = \log(x) + c$, where $c$ is a free parameter.}
	\label{fig:scheduling-time-empirical}
\end{figure}

\section{Further discussion}

\subsection{Comparison to barrier list schedulers}

Recall that one of the recent works on DAG scheduling with synchronization barriers is that of Papp \emph{et al.\@} ~\cite{Papp2024Efficient}, which analyzes schedulers not for a concrete application, more abstractly in terms of BSP cost. While the so-called BSPg scheduling heuristic in this work is rather different from our algorithm, the idea of prioritizing vertices in GrowLocal that are computable exclusively on a specific core was inspired by this algorithm. To show for completeness that our scheduler also significantly outperforms this BSPg algorithm, we also ran this BSPg scheduling algorithm as a baseline. The results show that GrowLocal achieves a factor $8.31\times$ geometric-mean speed-up to BSPg on the SutieSparse data set.

\subsection{On the synchronisation parameter $L$}

Recall that the parameter $L$ in GrowLocal, Algorithm~\ref{alg:GrowLocal}, represents the time cost of inserting a synchronization barrier, and is used to determine the parallelization rate in our GrowLocal algorithm.

If we consider the compute time of basic operations (additions or multiplications) with double precision numbers, and compare this to the time of synchronization, we get that the correct magnitude of $L$ ranges from a few hundreds to a few thousands on modern computing architectures. We ran some preliminary experiments with a few different choices of $L$ on this order of magnitude, and chose a value of $L=500$ based on these empirical observations.

\bibliography{references}
\bibliographystyle{alpha}

\end{document}